\documentclass[11pt]{article}

\usepackage{amsmath,amsfonts,amssymb,amsthm}
\usepackage{bm}
\usepackage{constants}
\usepackage[abspath]{currfile}
\usepackage[yyyymmdd,hhmmss]{datetime}
\usepackage[verbose,colorlinks]{hyperref}
\hypersetup{
  colorlinks,
  citecolor=violet,
  linkcolor=blue,
  urlcolor=red}
\usepackage[usenames,dvipsnames,svgnames,table]{xcolor}
\usepackage{graphicx}
\usepackage{cancel}
\usepackage{natbib}

\usepackage{accents}
\usepackage{enumitem}
\usepackage{subcaption}

\usepackage{tabularx} % automatic line break in table

\usepackage[margin=1in, centering]{geometry}
  
\theoremstyle{plain}
\newtheorem{assumption}{Assumption}
\newtheorem{theorem}{Theorem}

\newtheorem{proposition}{Proposition}
\newtheorem{lemma}{Lemma}
\theoremstyle{definition}

\newtheorem{remark}{Remark}

\DeclareMathOperator*{\argmin}{\operatorname{argmin}}

\DeclareMathOperator*{\ST}{\operatorname{ST}}

\newcommand {\BM} {\begin{bmatrix}}
\newcommand {\EM} {\end{bmatrix}}
\newcommand {\mc}[1]{{\mathcal{#1}}}
\newcommand {\mb}[1]{{\mathbb{#1}}}
\newcommand {\ms}[1]{{\mathsf{#1}}}

\newcommand {\Av}[1]{{\left| #1 \right|}}
\newcommand {\av}[1]{{| #1 |}}
\newcommand{\E}{\mb{E}}
\newcommand{\pr}{\mb{P}}
\newcommand {\Var}{\text{Var}}

\newcommand {\ip}[2]{{\left\langle #1, #2 \right\rangle}}
\newcommand {\set}[1]{{\left\{ #1 \right\}}}
\newcommand {\pa}[1]{{\left( #1 \right)}}
\newcommand {\br}[1]{{\left[ #1 \right]}}
\newcommand {\norm}[1]{{\lVert #1 \rVert}}
\newcommand {\Norm}[1]{{\left\lVert #1 \right\rVert}}

\newcommand {\tr}{{\,\ms{tr}}}
\newcommand {\card}{\ms{card}}

\newcommand {\sign}{\ms{sign}}

\newcommand{\fr}[1]{{\norm{ #1 }}}

\newcommand {\rank}{\ms{rank}}

\def \R {\mb{R}}

\def \X {\ms{X}}
\def \Z {\ms{Z}}

\newcommand{\for}{{\quad\mbox{for}\quad}} 

\newconstantfamily{FixedConstants}{symbol={M}}

\newcommand{\Ml}[1]{\ensuremath{\Cl[FixedConstants]{#1}}}
\newcommand{\Mr}[1]{\ensuremath{\Cr{#1}}}

\makeatletter
\newcommand*\DP{\mathpalette\DP@{.5}}
\newcommand*\DP@[2]{\mathbin{\vcenter{\hbox{\scalebox{#2}{$\m@th#1\bullet$}}}}}
\makeatother

\usepackage{algorithm}
\usepackage{algpseudocode}
\usepackage{booktabs}
\usepackage{multirow}

\def \v {\ms{v}}
\def \vh {\ms{vh}}
\def \c {\ms{c}}
\usepackage{authblk}

\begin{document}

\title{Covariance Regression based on Basis Expansion}

\author[1,2]{Kwan-Young Bak}
\author[1,2,*]{Seongoh Park}
\affil[1]{\normalsize School of Mathematics, Statistics and Data Science, Sungshin Women's University, 2, Bomun-ro 34da-gil, Seongbuk-gu, Seoul 02844, Republic of Korea}
\affil[2]{\normalsize Data Science Center, Sungshin Women's University, 2, Bomun-ro 34da-gil, Seongbuk-gu, Seoul 02844, Republic of Korea}
\affil[*]{\normalsize Corresponding Author: spark6@sungshin.ac.kr}

\date{\today}
\maketitle

\begin{abstract}
	This paper presents a study on an $\ell_1$-penalized covariance regression method. Conventional approaches in high-dimensional covariance estimation often lack the flexibility to integrate external information. As a remedy, we adopt the regression-based covariance modeling framework and introduce a linear covariance selection model (LCSM) to encompass a broader spectrum of covariance structures when covariate information is available. Unlike existing methods, we do not assume that the true covariance matrix can be exactly represented by a linear combination of known basis matrices. Instead, we adopt additional basis matrices for a portion of the covariance patterns not captured by the given bases. To estimate high-dimensional regression coefficients, we exploit the sparsity-inducing $\ell_1$-penalization scheme. Our theoretical analyses are based on the (symmetric) matrix regression model with additive random error matrix, which allows us to establish new non-asymptotic convergence rates of the proposed covariance estimator.
	The proposed method is implemented with the coordinate descent algorithm. We conduct empirical evaluation on simulated data to complement theoretical findings and underscore the efficacy of our approach. To show a practical applicability of our method, we further apply it to the co-expression analysis of liver gene expression data where the given basis corresponds to the adjacency matrix of the co-expression network.

\end{abstract}

\section{Introduction}
A covariance matrix stands as an essential component of multivariate analyses as it portrays the associations among variables. A pertinent example would be a co-expression analysis that explores the co-expression patterns among genes (\cite{Choi:2009,Rahmatallah:2013,Oh:2020}). To unravel the complex structures of data observed in the modern era, a high-dimensional covariance matrix has been extensively studied for decades. As a consequence, various estimators tailored to specific structural assumptions have been proposed; (adaptive) thresholding estimators (\cite{Bickel:2008a,Rothman:2009}), banding/tapering estimators (\cite{Bickel:2008b, Leng:2011}), factor models (\cite{Fan:2008,Fan:2013}), spiked models (\cite{Johnstone:2001}) to name a few.

However, existing works have an evident limitation in that external information (e.g. covariates, network structure) cannot be incorporated into the covariance model. To address this challenge, recent studies (\cite{zou2017covariance,Fan:2024}) proposed an attractive approach of regressing covariances using similarity matrices derived from given covariates as predictors. The authors adopted least square estimation or (quasi-) maximum likelihood estimation, and established related asymptotic properties. 
It is worth recognizing that the notion of regression modeling for a covariance matrix is not entirely new and closely related to the linear Gaussian covariance model (\cite{Anderson:1973,Zwiernik:2017}).
This approach considers the class of multivariate Gaussian distributions with linear constraints on the covariance matrix, represented as:
\[
Y \sim N_d(\mu, \Sigma), \quad \Sigma = \sum_{j=1}^J \alpha_j S_j,
\]
where $\{S_j\}_{j=1}^J \subset \mathbb{R}^{d \times d}$ are known symmetric matrices. The properties of the log-likelihood function and maximum likelihood estimator are discussed in \cite{Zwiernik:2017}. As a special case of the covariance regression model, \cite{lan2018covariance} utilized the network structure represented by an adjacency matrix $A$ of nodes as the regressors, which leads to the regression equation below.
\begin{equation}\label{eq:lan_model}
	\Sigma  = \alpha_0 I + \sum_{j=1}^s \alpha_j A^j.
\end{equation}
Here, the power of $A$ corresponds to the $j$-path between nodes, representing higher-order neighborhoods. It should be noted that this regression model implicitly assumes the population covariance matrix has the same eigen-space with the regressor matrix $A$. Though this approach is novel to decompose higher-order network effects on  covariance, the underlying assumption (\ref{eq:lan_model}) is greatly restrictive to cover a broad class of covariance matrices encountered in real-world scenarios. 

Motivated from it, we propose a linear covariance selection model (LCSM) to cover a larger class of covariance matrices when covariate information is available. In other words, our regression model is given as follows:
\[
\Sigma  = \alpha_0 I + \sum_{j=1}^s \alpha_j G_j + \Sigma_R,
\]
where an additional parameter $\Sigma_R$ is designed to capture the part not explained by the predetermined basis matrices $\set{G_j}_{j=1}^s$. To be specific, we express $\Sigma_R$ as a linear combination of basis matrices linearly independent of given basis matrices (see (\ref{eq:LCSM_model})). While this new class can cover any symmetric matrix, it can suffer from potential overparametrization in the high-dimensional space. To mitigate it, we select a linear covariance model adaptive to the data utilizing $\ell_1$-penalized least squares. 

This remainder term $\Sigma_R$ distinguishes our model from the previous works (\cite{zou2017covariance,lan2018covariance,Fan:2024}) wherein they only considered the model class of covariance matrices fully determined by the known basis matrices. Recently, \cite{Fan:2024} extended their earlier work by allowing the number of known basis matrices to be divergent in an order of $d^{\kappa}$ for some $\kappa< 1/4$ (see Condition 1 therein). In contrast, our approach accommodates a more general case ($\kappa \le 2$) through the regularization scheme.

%{\color{black}
This paper also presents a new analytic framework, a (symmetric) matrix regression model with additive random error matrix, under which we derive new non-asymptotic convergence rates of the proposed covariance estimator. This regression framework is particularly attractive as it no longer needs a likelihood function in estimation and allows a more general study in (symmetric) matrix regression.
%}

The structure of this paper is as follows. We introduce the class of the new covariance regression model and its estimation with details of the implementation in Section 2. Theoretical results are given in Section 3. Following \cite{lan2018covariance}, we focus on the covariance regression model based on an adjacency matrix and examine its empirical performance using simulated data in Section 4. In Section 5, we analyze the liver gene expression data in the context of a co-expression network analysis.  We conclude this paper with discussion in Section 6.

\section{Linear covariance selection model}

\subsection{Model and estimator}\label{ssec:model}

Let $\set{Y_i}_{i=1}^n$ be a set of centered multivariate observations with $Y_i = (Y_{ij}) \in \R^d$. For $j=1,\ldots, s$, let $A^j$ denote the $j$-th power of the adjacency matrix whose off-diagonal elements represent the number of $j$-paths between two nodes of a network. Given a set of predetermined basis matrices $\set{G_j}_{j=1}^s$, we consider a new covariance model:
\[
\Sigma =  \alpha_0 I + \Sigma_G + \Sigma_R
\]
where $I$ denotes the identity matrix of order $d$, $\Sigma_G = \sum_{j=1}^s \alpha_j G_j$, and $\Sigma_R$ is linearly independent of $I$ and $G_j$ for $j=1,\ldots, s$, where the linear independence implies that there exists no non-trivial linear combination of the matrices that equals the zero matrix. Here, $R$ represents the remaining portion of the covariance that is not explained by $G_1, \ldots, G_s$. 
Given the basis matrices $\set{G_j}_{j=1}^s$, we can expand the remainder matrix by $\Sigma_R=\sum_{j=1}^q \beta_j F_j$ where we choose $\set{F_k}_{k=1}^q$ such that $\ip{G_j}{F_k} = 0$ for $j=1,\ldots, s$ and $k=1,\ldots, q$. Here, $\ip{D_1}{D_2}$ denote the Frobenius inner product $\tr(D_1^\top D_2)$ between two matrices $D_1$ and $D_2$. The choice of the basis set $\set{F_k}_{k=1}^q$ will be discussed in Section \ref{ssec:implementation}.
%
%
%Taking the expectation in both sides, we get $\Sigma  \triangleq \Cov(Y) =  \alpha_0 I + \Sigma_G + \Sigma_R$, which is assumed to be positive definite,

The (true) covariance can be re-expressed as
\begin{equation}\label{eq:LCSM_model}
	\Sigma = \alpha_0^* I + \sum_{j=1}^s \alpha_j^* G_j + \sum_{j=1}^q \beta_j^* F_j.	
\end{equation}
We enumerate the basis as 
\[
\set{I, G, \ldots, G_s, F_1, \ldots, F_q} = \set{B_j}_{j=1}^p
\] with $p=s + q + 1$ and define 
\[
\theta^* = (\theta_j^*) = (\alpha_0^*, \alpha_1^*, \ldots, \alpha_s^*, \beta_1^*, \ldots, \beta_q^*) \in \R^p.
\] 
When considering a saturated model to represent an arbitrary symmetric matrix, $p$ is determined to be $d(d+1)/2$.
For $\theta = (\theta_1,\ldots, \theta_p) \in \R^p$, we denote the linear predictor by
\[
\Gamma_\theta = \theta_1 B_1 + \cdots + \theta_p B_p
\]
so that $\Sigma = \Gamma_{\theta^*}$. 

Since $\E [Y_i Y_i^\top] = \Sigma$, the empirical risk is defined as
\[
\ell(\theta) = \sum_{i=1}^n \fr{Z_i - \Gamma_\theta}^2,
\]
where $Z_i = Y_i Y_i^\top$ and $\fr{\cdot}=\ip{\cdot}{\cdot}$ denotes the Frobenius norm of a matrix.
Adopting an $\ell_1$-penalty to impose sparsity structure, we define the objective function to minimize by
\[
\ell^\lambda(\theta) = \ell(\theta) + 2\lambda \sum_{j=2}^p \av{\theta_j},
\]
where $\lambda > 0$ is the complexity parameters. The summation begins from $j=2$ since the intercept term is not penalized in general.
Alternatively, we may consider imposing the $\ell_1$-penalty only on coefficients corresponding to the remainder.

Now we define an estimator:
\begin{equation}\label{eq:LCSM_estimator}
	\hat{\theta}^\lambda = \argmin_{\theta \in \R^p} \ell^\lambda (\theta).	
\end{equation}
Unless absolutely necessary, we omit the superscript for notations to remain uncluttered. The covariance estimator is given by
\[
\hat{\Sigma} = \Gamma_{\hat{\theta}}.
\]
The existence and uniqueness of $\hat{\theta}$ are established in Lemma \ref{le:exist-unique} in \ref{sec:proofs}. If $\hat{\Sigma}$ does not satisfy positive definiteness, a slight adjustment is made to obtain a positive definite estimator, as described in Section \ref{ssec:spd-adjust}.

\subsection{Implementation scheme}\label{ssec:implementation}

\subsubsection*{Choice of the basis set}

To construct the basis matrices $\set{F_k}_{k=1}^q$ for the remainder matrix, we start with the matrix of the form
\[
\mc{D} = \BM \vh(I) & \vh(G_1) \cdots \vh(G_s) \EM \in \R^{d(d+1)/2 \times (s+1)},
\]
where $\vh(\cdot):\R^{d \times d} \to \R^{d(d+1)/2}$ is the half-vectorization operator that vectorizes the lower triangular part of a symmetric matrix. The half-vectorization operator is used only for basis generation purposes, but not used in estimation or theoretical analysis. Then, we can find $d(d+1)/2 - s - 1$ orthogonal vectors in the left null space of $\mc{D}$. We note that these vectors are orthogonal to the column space of $\mc{D}$ by construction. Now taking the inverse operation $\vh^{-1}(\cdot)$ yields the basis matrices $\set{B_j}_{j=s+2}^p$ in $\mathbb{R}^{d \times d}$. In particular, we consider an orthonormal basis $\set{F_k}_{k=1}^q$ with $\ip{F_k}{F_l}=\delta_{kl}$.

Remark that choosing another set $\big\{\tilde{F}_k\big\}_{k=1}^q$ of basis matrices orthogonal to $I$ and $\{G_j\}_{j=1}^s$ does not change the estimation of $\{\alpha_j\}_j$ due to the orthogonality. We also note that the corresponding coefficients $\tilde{\beta}_1, \ldots, \tilde{\beta}_p$ may have different values, but $\Sigma_R$ is still identifiable.

%
%In the analyses conducted in Section \ref{sec:simulation} and \ref{sec:realdata}, we have found them by investigating the left null space of the matrix whose columns are given by the half-vectorized basis matrices $\set{G_j}_{j=1}^s$ 

\subsubsection*{Coordinate descent algorithm}
This section presents an implementation algorithm to calculate the estimator $\hat{\theta}$. 
We adopt a coordinate descent algorithm to obtain $\hat{\theta}$. One may refer to, for example, \cite{friedman2007pathwise,tseng2001convergence} for the applications of the coordinate descent algorithm in statistics and its convergence properties, respectively. Let $(\tilde{\theta}_1, \ldots, \tilde{\theta}_p)$ denote the current values of the coefficients. For $j=1,\ldots, p$, the univariate objective function of the $j$-th coefficient, holding the other coefficients fixed, has the form
\[
\ell^\lambda_j(x) = \ell^\lambda(\tilde{\theta}_1, \ldots, \tilde{\theta}_{j-1}, x, \tilde{\theta}_{j+1}, \ldots, \tilde{\theta}_p).
\]
The update formula is given by
\[
\tilde{\theta}_j \leftarrow \argmin_{x \in \R} \ell_j^\lambda (x) \for j=1,\ldots, p.
\]
Let $D_{kl}$ denote the $(k,l)$ element of an arbitrary matrix $D$. 
An explicit from of the update formula is given in Proposition \ref{prop:update-formula}.

\begin{proposition}\label{prop:update-formula}
	The minimizer of the univariate objective function $\ell_j^\lambda$ is given by
	\[
	\argmin_{x \in \R} \ell_j^\lambda (x)
	= 
	\ST \pa{\frac{\sum_{i=1}^n \sum_{k,l=1}^d r_{ijkl} B_{j,kl}}{n \sum_{k,l=1}^d B_{j,kl}^2},  \frac{\lambda}{n \sum_{k,l=1}^d B_{j,kl}^2}} \for j=1,\ldots, p,
	\]
	where $r_{ijkl} = Z_{i,kl} - \sum_{j' \ne j} \tilde{\theta}_{j'} B_{j',kl}$ and the soft-thresholding operator is defined as 
	\[ 
	\ST(a, \lambda) 
	=
	\begin{cases}
	a - \lambda &\mbox{when $a>\lambda$}\\
	0 &\mbox{when $\av{a} \le \lambda$}\\
	a + \lambda &\mbox{when $a < -\lambda$}
	\end{cases}.
	\]
\end{proposition}

The coordinate descent algorithm cyclically updates the coefficients accordingly until convergence. The iteration stops when the $\ell_2$-difference of the coefficient values from the previous iteration and the updated values is less than a pre-specified tolerance level, say $10^{-6}$.

\subsubsection*{Positive definiteness of the covariance estimator}\label{ssec:spd-adjust}

%% asymptotic 
We give a brief remark on the positive definiteness of the estimator $\hat{\Sigma}$.
Consider the following parameter space:
\[
\Theta = \set{\theta: \Gamma_\theta \succ 0},
\]
where we denote $A_1 \succ A_2$ if $A_1 - A_2$ is positive definite. Suppose that $\theta^* \in \Theta = \set{\theta: \Gamma_\theta \succ 0}$. We note that $\Theta$ is non-empty since it contains $\theta$ with $\theta_1 >0$ and $\theta_j = 0$ for $j=1,\ldots, p$. Proposition $1$ of \cite{zou2017covariance} implies that the parameter space $\Theta$ is an open set under Assumption \ref{cond:B2-bd}. It follows that the unconstrained estimator $\hat{\theta}$ coincides with the the constrained estimator $\hat{\theta}_+$ in the asymptotic sense provided that $\hat{\theta}$ is a consistent estimator of $\theta^* \in \Theta$.

%% finite sample
However, in the finite sample case, the covariance estimator based on $\hat{\theta}$ may not be positive definite. In such cases, slight modifications can be made to obtain a positive definite estimator $\hat{\Sigma}_+$ close to $\hat{\Sigma}$. For example, one can adopt the alternating direction method of multiplier to solve a constrained problem as in \cite{xue2012positive,zou2017covariance}.  
Another approach involves performing the spectral decomposition of $\hat{\Sigma}$ and then approximating the nearest positive definite estimator in terms of the Frobenius metric by using only the positive eigenvalues.
Such algorithms may come with the drawback of imposing an additional computational burden or making the contribution of basis matrices less transparent.
In the majority of applications for the proposed covariance regression method, the coefficient of the identity matrix does not carry practical significance. Thus, we suggest a method of adding $\omega>0$ to the diagonal elements of $\hat{\Sigma}$ with $\omega$ being greater than the smallest eigenvalue of non-positive definite $\hat{\Sigma}$, which ensures the positive definiteness of the resulting covariance estimator. Through the results of the numerical studies in Section \ref{sec:simulation} and Section \ref{sec:realdata},  we found that such corrections are rarely needed. Furthermore, by choosing small $\omega$, we can guarantee that the difference between the Frobenius distance of the corrected estimator and the true covariance, and the distance between the initial estimator and the true covariance, is small enough.

\subsubsection*{Selection of the optimal complexity parameter}
The estimates are computed for an increasing sequence of the complexity parameters $\lambda_1 < \cdots < \lambda_m$, where $\lambda_m$ is sufficiently large. A possible choice of $\lambda_m$ is given by Proposition \ref{prop:lambda-max}. The result implies that we can find the maximum candidate value of $\lambda$ based on the Karush-Kuhn-Tucker (KKT) condition and thus the range $\{\lambda > \lambda_m\}$ no longer needs to be considered. 

\begin{proposition}\label{prop:lambda-max}
	Denote by $\hat{\theta}^\lambda$ be the minimizer of $\ell^\lambda$ associated with the complexity parameter $\lambda$. Define 
	\[
	\lambda_m = \max_{1 \le j \le p} \Av{\sum_{i=1}^n \ip{B_j}{Z_i}}.
	\]
	Then, we have $\hat{\theta}^\lambda = 0$ for $\lambda > \lambda_m$.
\end{proposition}

Based on the result of Proposition \ref{prop:lambda-max}, $\lambda_m$ is determined, and then $\lambda_1$ is determined by the relation $\lambda_1 = \delta \lambda_m$, where $\delta$ is a small number. Then, the entire sequence of $m$ complexity parameters is generated by the evenly spaced points between $\lambda_1$ and $\lambda_m$ on the log scale. The estimator $\hat{\theta}^\lambda$ is computed along a sequence of complexity parameters $\set{\lambda_1, \ldots, \lambda_m}$. To improve the efficiency of an algorithm, we adopt a warm-start strategy \citep{friedman2007pathwise} in which each solution of the optimization problem for a given complexity parameter is used as an initial value for the next problem.

To choose the optimal complexity parameter $\lambda_\text{opt}$, we use a variant of the Akaike information criterion (AIC); see \cite{akaike1974new}. The AIC for $\lambda_l$ is defined as
\[
\text{AIC}_l =  \ell(\hat{\theta}^{\lambda_l}) + 2\, \card\set{S(\hat{\theta}^{\lambda_l})} / d,
\]
where
\[
S(\theta) = \set{j \colon \theta_j \ne 0,\ j=1,\ldots, p} \for \theta \in \R^p,
\]
and $\card\set{\cdot}$ denotes the cardinality of a set.
The optimal value for $\lambda$ is chosen as $\lambda_\text{opt} = \lambda_{\hat{l}}$, where $\hat{l} = \argmin_{1 \le l \le m} \text{AIC}_l$.

\section{Theoretical results}\label{sec:theory}

\subsection{Assumptions}\label{ssec:assumption}
This section presents assumptions to derive theoretical results on the convergence of the proposed estimator.
For a theoretical framework, we consider a more general matrix regression model 
\begin{equation}\label{eq:reg_model}
	Z_i = \alpha_0 I + \Sigma_G + \Sigma_R + \varepsilon_i \for i=1,\ldots, n,
\end{equation}
where $Z_i$ is a symmetric random matrix and $\varepsilon_i \in \R^{d \times d}$ are independent random matrices with $\E[\varepsilon_i] = 0$. Our case is a special case of \eqref{eq:reg_model}, i.e. $Z_i = Y_i Y_i^\top$. This regression model is useful because the estimation process does not depend on the likelihood.
Since the empirical risk is the sum of squares, one may have an urge to apply the standard lasso theory (see, for example, \cite{buhlmann2011statistics}) to $\vh(Y_i Y_i^\top)$. However, this approach only works under the independence assumption on the different (upper diagonal) components, which is not valid for the covariance regression. 
As a remedy, we characterize the tail behavior of the random error matrix in terms of the Laplace transform. In this paper, we adopt the sub-Gaussian or sub-exponential tail assumptions on the behavior of the random error. One may refer to \cite{tropp2012user} and \cite{wainwright2019high} for various probability inequalities based on a matrix generalization of the moment generating function under the assumptions.

We briefly provide preliminary notions. 
Given a symmetric random matrix $D$ in $\R^{d \times d}$, we use the matrix exponential map $e^D=\sum_{k=0}^\infty D^k / k!$ to define the moment generating function
\[
m_D(t) = \E[e^{tD}] = \sum_{k=0}^\infty \frac{t^k}{k!} \E[D^k],
\]
provided that the polynomial moments exist.
A zero-mean random matrix $D$ is said to be sub-Gaussian with positive definite matrix $W$ being the parameter if 
\[
m_D(t) \preceq e^{\frac{t^2 W}{2}} \for t \in \R,
\]
where we denote $A_1 \preceq A_2$ if $A_2 - A_1$ is positive semi-definite. 
When $D$ satisfies
\[
m_D(t) \preceq e^{\frac{t^2 W}{2}} \for \av{t} \le 1/a,
\]
it is said to be sub-exponential with parameters $(W,a)$. We note that any sub-Gaussian random matrix is sub-exponential with parameters $(W, 0)$.

Let $M_1$, $M_2$, and $M_3$ be fixed positive constants independent of $n$. We make the following assumption on the basis matrices.
\begin{assumption}\label{cond:B2-bd}
	$\max_{1 \le j \le p} \fr{B_j} \le \Ml{B2-bd}$.
\end{assumption}
\noindent
Assumption \ref{cond:B2-bd} implies that the basis matrices are normalized in terms of the Frobenius norm. This can be easily satisfied via standardization of covariates, which is commonly used in penalized regression. Tail behavior of the random error matrix is characterized by one of the following assumptions.
\begin{assumption}\label{cond:sub-Gaussian}
	For $i=1,\ldots, n$, the independent zero-mean random errors $\varepsilon_i$ satisfy 
	\[
	m_{\varepsilon_i}(t) \preceq e^{\frac{t^2 W}{2}} \for t \in \R.
	\]
\end{assumption}
\begin{assumption}\label{cond:Bernstein-moments}
	For $i=1,\ldots, n$, the independent zero-mean random errors $\varepsilon_i$ satisfy 
	\[
	\E[\varepsilon_i^k] \preceq \frac{1}{2}k! b^{k-2} \Var(\varepsilon_i),\ k=2,3,\ldots
	\]
	for some $b > 0$, where $\varepsilon_i^k$ is the $k$th power of the matrix $\varepsilon_i$ and $\Var(\varepsilon_i)$ is defined as $\E[\varepsilon_i^2]$.
\end{assumption}

\begin{remark}
	Assumption \ref{cond:sub-Gaussian} is the sub-Gaussian condition that covers a wide class of the random error. For example, if $\psi$ is a sub-Gaussian real-valued random variable and $D$ is a fixed matrix, or $\psi$ is a Rademacher variable and $D$ is a random matrix with a bounded spectral norm, $\psi D$ satisfies Assumption \ref{cond:sub-Gaussian}.
\end{remark}

\begin{remark}
	Assumption \ref{cond:Bernstein-moments} is also known as Bernstein moment condition that implies the sub-exponential condition. Various classes of random matrices belong to the sub-exponential class including the ones bounded in the spectral norm; see, for example, \cite{tropp2012user}. In particular, \cite{zhu2012short} showed that the Bernstein condition is satisfied for Wishart distribution. Thus, this condition covers the case of multivariate Gaussian observations.
	Moreover, if $Y$ is sub-Gaussian as a random vector with $\mathbb{E} YY^\top =\Sigma$, then $YY^\top - \Sigma$ also meets a similar condition. Our analytical framework includes this case with a slight modification in the proof; see Lemma \ref{le:subexp_YYt} and Lemma \ref{le:bernstein-ineq-subGaussian} in \ref{sec:proofs}.
	%. The proofs are postponed until Appendix \ref{ssec:verify_bernstein}.
\end{remark}

\subsection{Results}\label{ssec:theory_result}

We consider the complexity parameter used in the tuning parameter:
\[
\tau(\nu) = \sqrt{\frac{(u_p)^2 \Mr{B2-bd}^2 \log(2d / \nu)}{n}},
\]
where $\textcolor{black}{u_p} = \underset{1 \le j \le p}{\max}\sqrt{\rank(B_j)}$,
and $0 < \nu < 1$ is a user-defined constant. Then, when $\set{\varepsilon_i}_{i=1}^n$ satisfy Assumption \ref{cond:sub-Gaussian} with parameters $\set{W_i}_{i=1}^n$ with $\textcolor{black}{\sigma_{W,n}^2} = \frac{1}{n} \norm{\sum_{i=1}^n W_i}_2$, we suggest using
\begin{equation}\label{eq:lambda-gauss}
	\lambda = \sqrt{2} \sigma_{W,n} n \tau(\nu).
\end{equation}
For the sub-exponential case, where $\textcolor{black}{\sigma^2_{\varepsilon,n}} = \norm{\sum_{i=1}^n \Var(\varepsilon_i)}_2 / n$, we propose to use
\begin{equation}\label{eq:lambda}
	\lambda = n \pa{\sqrt{2} \sigma_{\varepsilon,n} \tau(\nu) + 2b\{\tau(\nu)\}^2}.
\end{equation}

Theorem \ref{th:nonasymptotic-bd-gauss} and Theorem \ref{th:nonasymptotic-bd} provide non-asymptotic upper bounds for the discrepancy between the estimator and underlying covariance, where the Frobenius norm is used
as a measure for goodness-of-fit.
\begin{theorem}\label{th:nonasymptotic-bd-gauss}
	Suppose Assumption \ref{cond:B2-bd} and \ref{cond:sub-Gaussian} hold and choose $\lambda$ as \eqref{eq:lambda-gauss}. Then, with probability at least $1-\nu$, we have
	\[
	\fr{\hat{\Sigma} - \Sigma }^2 \le 
	4\sqrt{2} \sigma_{W,n} u_p \av{\theta^*}_1 \sqrt{\frac{\Mr{B2-bd}^2 \log(2d / \nu)}{n}},
	\]
	where $\av{\cdot}_1$ is the $\ell_1$ norm of a Euclidean vector.
\end{theorem}

\begin{theorem}\label{th:nonasymptotic-bd}
	Suppose Assumption \ref{cond:B2-bd} and \ref{cond:Bernstein-moments} hold and choose $\lambda$ as \eqref{eq:lambda}. Then, with probability at least $1-\nu$, we have
	\[
	\fr{\hat{\Sigma} - \Sigma }^2 \le 
	4\sqrt{2} \sigma_{\varepsilon,n} u_p \av{\theta^*}_1 \sqrt{\frac{\Mr{B2-bd}^2 \log(2pd / \nu)}{n}} +
	8b (u_p)^2 \av{\theta^*}_1 \frac{\Mr{B2-bd}^2 \log(2d / \nu)}{n}.
	\]
\end{theorem}

\begin{remark}
	The upper bounds in Theorem \ref{th:nonasymptotic-bd-gauss} and Theorem \ref{th:nonasymptotic-bd} are non-asymptotic results, and $\nu$ represents the trade-off between the risk upper bound and the probability bound.
\end{remark}

%Theorem \ref{rem:asymp-rate} is an immediate consequence of Theorem \ref{th:nonasymptotic-bd-gauss} and Theorem \ref{th:nonasymptotic-bd}. 
%We adopt the standard Vinogradov asymptotic notations, $a_n \ll b_n$ as $n \to \infty$ for two positive real sequences $\set{a_n}$ and $\set{b_n}$.

%\begin{theorem}\label{rem:asymp-rate}
%Assume that $d = d(n) \to \infty$ as $n \to \infty$. The following asymptotic results hold with probability tending to $1$ as $n \to \infty$.

A couple of comments are in order about the results of  Theorem \ref{th:nonasymptotic-bd-gauss} and Theorem \ref{th:nonasymptotic-bd}. 
{\color{black}
	\begin{remark}\label{rem:asymp-rate}
		Suppose the rank of the predetermined basis is low and $\beta^*$, the coefficients for the remainder matrices, is sparse. Moreover, the basis matrices that consist of the remainder matrix can be taken to be low-rank (see a rank-one basis problem, e.g. \cite{Nakatsukasa:2017}). Then, $u_p \av{\theta^*}_1$ does not depend on $d$. In this case, we can derive the following asymptotic results with probability tending to $1$ as $d = d(n) \to \infty$, $n \to \infty$.
		\begin{itemize}
			\item Suppose Assumption \ref{cond:B2-bd} and \ref{cond:sub-Gaussian} hold and $\lambda$ is chosen as \eqref{eq:lambda-gauss}. Then, we have
			%\red{
			\[
			\fr{\hat{\Sigma} - \Sigma }^2 \le \Ml{rate-gauss} u_p \av{\theta^*}_1 \sqrt{\frac{\log d}{n}}.
			\]
			%}
			
			\item 
			Suppose Assumption \ref{cond:B2-bd} and \ref{cond:Bernstein-moments} hold and $\lambda$ is chosen as \eqref{eq:lambda}. Then, we have
			%\red{
			\[
			\fr{\hat{\Sigma} - \Sigma }^2 \le \Ml{rate} u_p \av{\theta^*}_1 \max\set{\sqrt{\frac{\log d}{n}} ,\frac{ {u_p b} \log d}{n}}.
			\]
			%}
		\end{itemize}
		
	\end{remark}
}

%\begin{remark}
%The condition $\av{\theta^*}_1 \log d \ll n$ is likely to hold in a broad class of covariance estimation problems, especially when the underlying covariance is sparsely supported, thereby $\av{\theta^*}_1$ being close to a constant.
%\end{remark}

{
	\begin{remark}\label{rem:}
		%	{\color{red}However, $b=\text{tr}(\Sigma)$ for $\epsilon \sim N(0, \Sigma)$, and thus $b=O(d)$ could affect the convergence rate?}
		\cite{zou2017covariance} showed that $n \fr{\hat{\Sigma} - \Sigma }^2$ converges to some quadratic form in distribution under some eigenvalue conditions. Seemingly, the asymptotic convergence rate given in Remark \ref{rem:asymp-rate} is sub-optimal, but there is a critical difference from the previous result. The regression framework \eqref{eq:reg_model} covers much more general error structures, and the error matrix in \eqref{eq:reg_model} has $O(d^2)$ random components. However, the matrix $Z_i = Y_i Y_i^\top$ from the multivariate normal $Y_i$ only requires to handle $O(d)$ components. Consequently, the proposed analytic framework allows a more general study in matrix regression in exchange for sub-optimal rate in this specific case.
	\end{remark}
}

%{\color{red} valid?
%\begin{remark}\label{rem:generalization}
%We do not use the orthonormality of basis matrices in the remainder part and establish the estimation consistency based solely on conservative bounds. However, we believe 
%
%We will leave the proof of Lemma \ref{le:ip} in \ref{sec:proofs} in its current form, which can be used to achieve tighter convergence when upper bounds for the nuclear norm, Frobenius norm, maximum element, etc., of the basis matrices are provided.  (will be modified)
%\end{remark}
%}

\section{Simulation study}\label{sec:simulation}

We conduct numerical studies based on simulated data to demonstrate the finite
sample performance of the proposed method. To imitate the co-expression analysis, we construct the basis matrices using network information represented by an adjacency matrix.

\subsection{Data generation}
We consider three types of adjacency matrices and examine $18$ distinct scenarios depending on the number of variables $d$ and the highest order $s$ of the adjacency matrices. For the numerical studies, we set $n=50$ and consider $d=20, 50, 80$ and $s = 2, 3$ for each adjacency matrix type. The tuning of the complexity parameter $\lambda$ was performed using the adjusted AIC, as explained in Section \ref{ssec:implementation}. The estimation performance of the proposed estimator is compared with that of the covariance regression estimator proposed by \cite{lan2018covariance}, titled by LCM. Since the core of the covariance regression model lies in estimating and interpreting the coefficients to understand the portion of the covariance structure attributed to the given basis matrix, we do not consider other methods, such as thresholding-based covariance estimation, in the comparison.

The data generation scheme is based on the regression model \eqref{eq:reg_model}
\[
Z_i = \Sigma_A + \Sigma_R + \varepsilon_i \for i=1,\ldots, n,
\]
where $\Sigma_A$ is the covariance matrix associated with the adjacency matrix and $\varepsilon_i$ is from centered Wishart distribution, i.e., $\varepsilon_i = \tilde{\varepsilon}_i / d - \Sigma_e$, where $\tilde{\varepsilon}_i \sim W_d(\Sigma_e, d)$. We set $\Sigma_e = \sigma_e^2 \text{I}$, $\sigma_e^2=1$, and $\Sigma_A = \alpha_0^* I + \sum_{j=1}^s \alpha_j^* A^j$, where $A$ is an adjacency matrix.
We consider three types of the adjacency matrix (\cite{Tan:2014}):
\begin{itemize}
	\item \textbf {(Type 1) Network with hub nodes:} Generate a lower triangular matrix with non-hub nodes having a connection probability of $0.02$ and hub nodes having a connection probability of $0.7$. The adjacency matrix is created by symmetrizing it. The number of hub nodes is set to $2,3,4$ for $d = 20, 50, 80$, respectively.
	\item \textbf {(Type 2) Network with two connected components and hub nodes:} Two adjacency matrices are generated as in Type 1 of dimension $d/2$. Then, connect them in a block-diagonal form to create the adjacency matrix. The number of hub nodes is set to $(1,1)$, $(1,1)$, $(2,2)$  for $d = 20, 50, 80$, respectively.
	\item \textbf {(Type 3) Random graph:} Generate a lower triangular matrix with each node having a connection probability $0.2$. The adjacency matrix is created by symmetrizing it.
\end{itemize}
The remainder covariance matrix $\Sigma_R = \sum_{j=1}^q \beta_j^* F_j$, where $\{F_j\}_{j=1}^q$ are the matrix basis orthogonal to $\{I, A, \ldots, A^s\}$. We set $\alpha^*=10 * \boldsymbol{1}_s$ with $\boldsymbol{1}_s \in \R^s$ being the vector of ones, 	$\beta^* = 50 * (1, -1, -1, 1, 0, \ldots, 0)$, and $\alpha_0^* = 300$. We set $\alpha_0^*$ relatively large to easily ensure  positive definiteness in the covariance generation for simulation. It does not have much impact on measuring the overall estimation performance.

\subsection{Results}

% latex table generated in R 4.4.1 by xtable 1.8-4 package
% Wed Dec 11 13:24:03 2024
\begin{table}[H]
	\centering
		\begin{tabular}{lcccc}
			\hline
			Adj. Type & Dimension & Adj. Order & LCSM & LCM \\ 
			\hline
			Type 1 & d = 20 & s = 2 & 18.984 (0.017) & 40.125 (0.023) \\ 
			&  & s = 3 & 18.998 (0.013) & 40.131 (0.024) \\ 
			& d = 50 & s = 2 & 48.964 (0.012) & 53.828 (0.010) \\ 
			&  & s = 3 & 49.005 (0.007) & 53.865 (0.007) \\ 
			& d = 80 & s = 2 & 78.963 (0.010) & 80.918 (0.009) \\ 
			&  & s = 3 & 79.011 (0.009) & 80.965 (0.008) \\ 
			\hline
			Type 2 & d = 20 & s = 2 & 18.995 (0.012) & 40.132 (0.021) \\ 
			&  & s = 3 & 18.997 (0.011) & 40.132 (0.024) \\ 
			& d = 50 & s = 2 & 48.983 (0.009) & 53.845 (0.008) \\ 
			&  & s = 3 & 49.000 (0.004) & 53.860 (0.005) \\ 
			& d = 80 & s = 2 & 78.976 (0.009) & 80.930 (0.008) \\ 
			&  & s = 3 & 79.001 (0.003) & 80.954 (0.003) \\ 
			\hline
			Type 3 & d = 20 & s = 2 & 18.921 (0.021) & 40.097 (0.012) \\ 
			&  & s = 3 & 19.089 (0.030) & 40.170 (0.015) \\ 
			& d = 50 & s = 2 & 48.952 (0.019) & 53.823 (0.008) \\ 
			&  & s = 3 & 49.205 (0.020) & 53.981 (0.011) \\ 
			& d = 80 & s = 2 & 78.970 (0.005) & 80.922 (0.004) \\ 
			&  & s = 3 & 79.706 (0.019) & 81.253 (0.017) \\ 
			\hline
		\end{tabular}
	\caption{The average FE values and the standard errors (in parentheses) of LCSM (proposed method) and LCM with $1000$ replications simulation.}
	\label{tab:sim-fe}
\end{table}
% beta MSE result
% latex table generated in R 4.3.2 by xtable 1.8-4 package
% Thu Nov 23 12:53:43 2023
\begin{table}[H]
	\centering
		\begin{tabular}{lcccc}
			\hline
			Adj. Type & Dimension & Adj. Order & LCSM & LCM \\
			\hline
			Type 1 & d = 20 & s = 2 & 45.063 (0.044) & 1607.577 (0.044) \\ 
			&  & s = 3 & 45.299 (0.187) & 1434.055 (0.188) \\ 
			& d = 50 & s = 2 & 299.365 (0.009) & 1861.868 (0.008) \\ 
			&  & s = 3 & 267.020 (0.007) & 1655.998 (0.009) \\ 
			& d = 80 & s = 2 & 778.805 (0.012) & 2341.319 (0.012) \\ 
			&  & s = 3 & 694.154 (0.016) & 2083.157 (0.017) \\ 
			\hline
			Type 2 & d = 20 & s = 2 & 49.103 (0.291) & 1611.617 (0.291) \\ 
			&  & s = 3 & 54.114 (0.295) & 1442.932 (0.298) \\ 
			& d = 50 & s = 2 & 299.798 (0.005) & 1862.308 (0.006) \\ 
			&  & s = 3 & 267.543 (0.068) & 1656.407 (0.067) \\ 
			& d = 80 & s = 2 & 779.327 (0.009) & 2341.830 (0.010) \\ 
			&  & s = 3 & 693.563 (0.005) & 2082.477 (0.005) \\ 
			\hline
			Type 3 & d = 20 & s = 2 & 44.430 (0.006) & 1606.954 (0.006) \\ 
			&  & s = 3 & 42.029 (0.028) & 1431.073 (0.027) \\ 
			& d = 50 & s = 2 & 299.223 (0.011) & 1861.812 (0.006) \\ 
			&  & s = 3 & 285.223 (0.089) & 1664.650 (0.045) \\ 
			& d = 80 & s = 2 & 779.333 (0.005) & 2341.666 (0.004) \\ 
			&  & s = 3 & 950.279 (0.063) & 2121.770 (0.083) \\ 
			\hline
		\end{tabular}
	\caption{The average MSE values for coefficients estimation and the standard errors (in parentheses) of LCSM (proposed method) and LCM with $1000$ replications simulation.}
	\label{tab:sim-beta}
\end{table}

\begin{figure}[H]
	\begin{subfigure}{\textwidth}
		\centering
		\includegraphics[width=.8\linewidth]{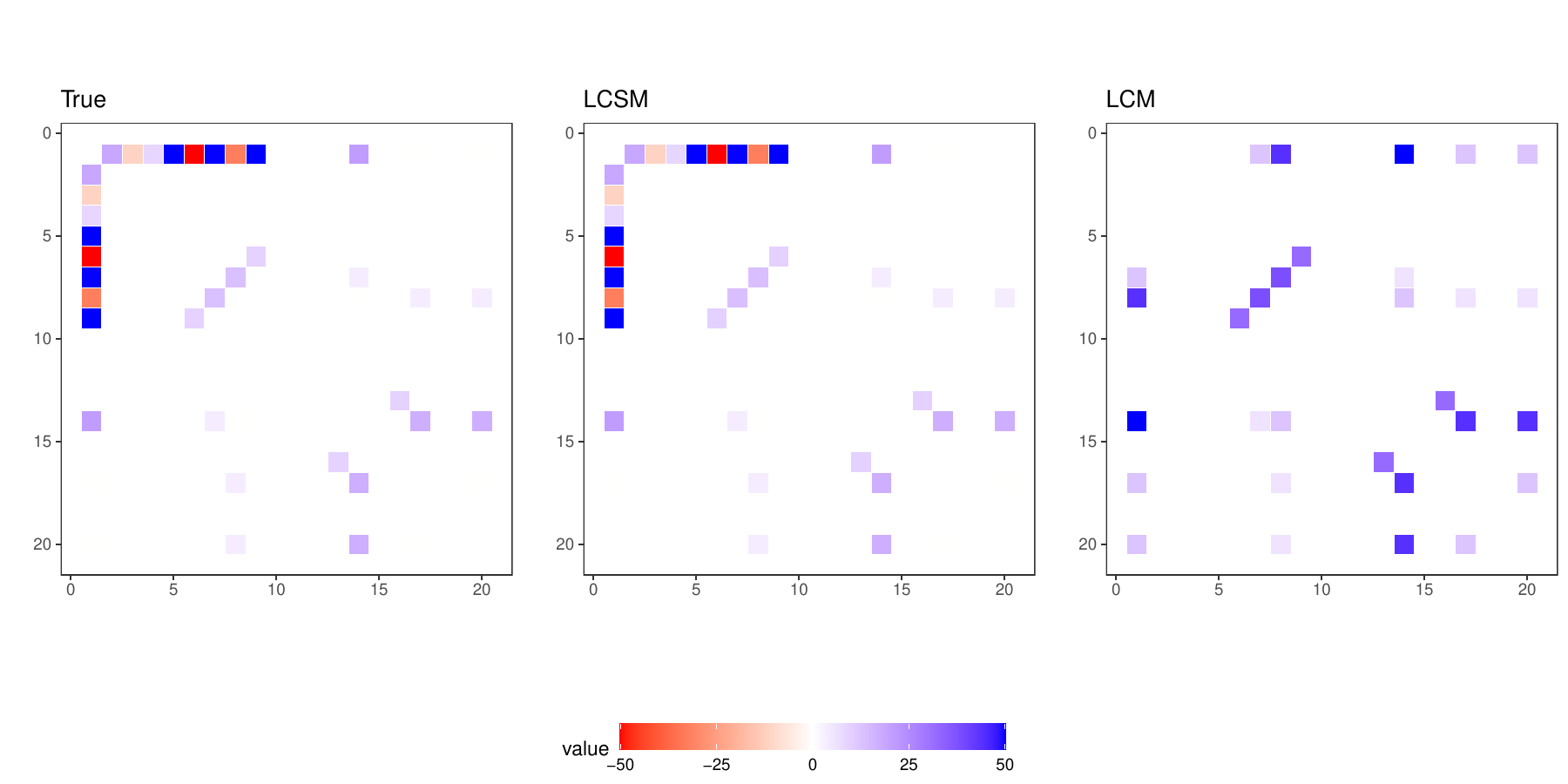}
	\end{subfigure}%
	\\
	\begin{subfigure}{\textwidth}
		\centering
		\includegraphics[width=.8\linewidth]{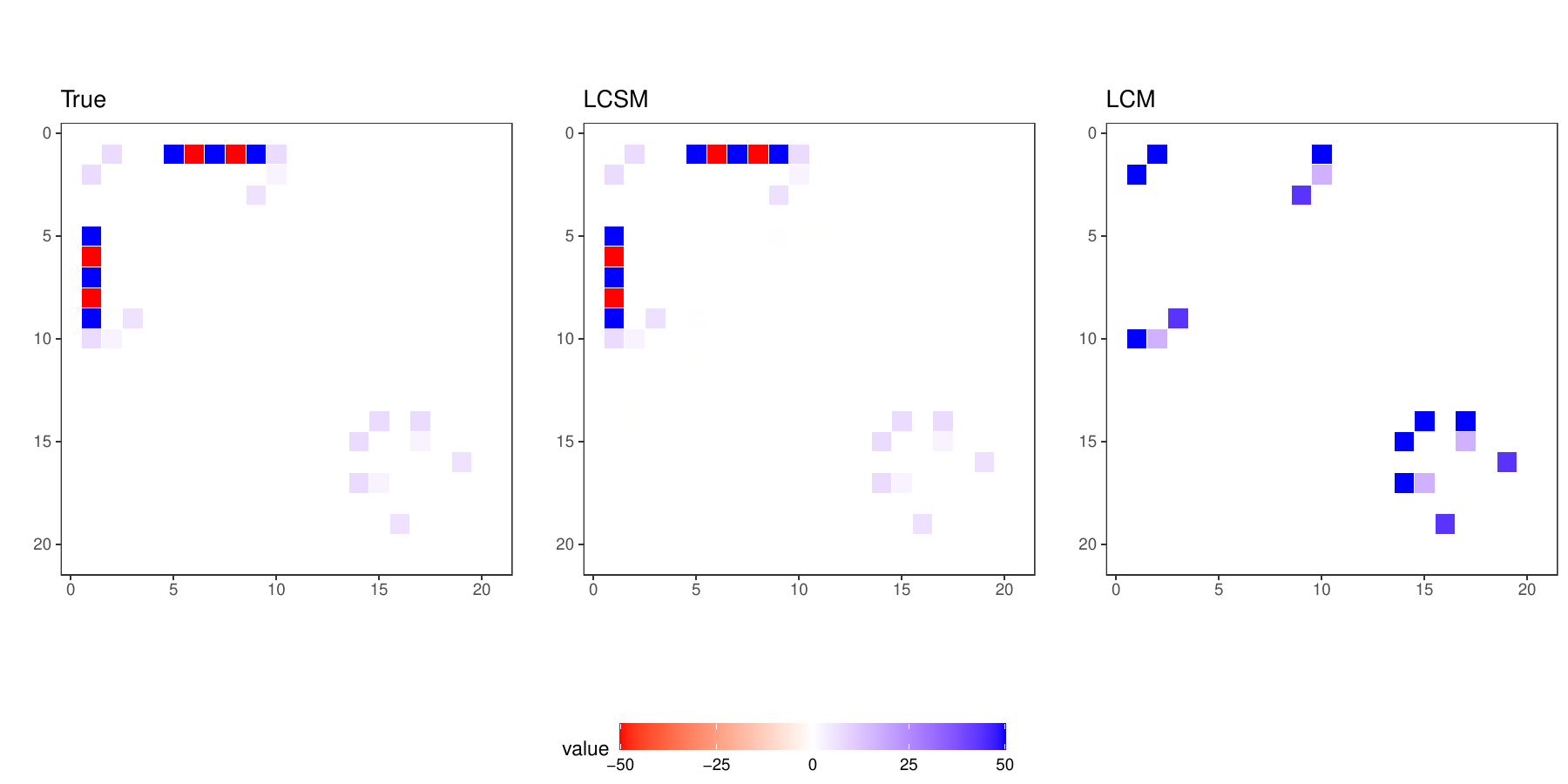}
	\end{subfigure}%
	\\
	\begin{subfigure}{\textwidth}
		\centering
		\includegraphics[width=.8\linewidth]{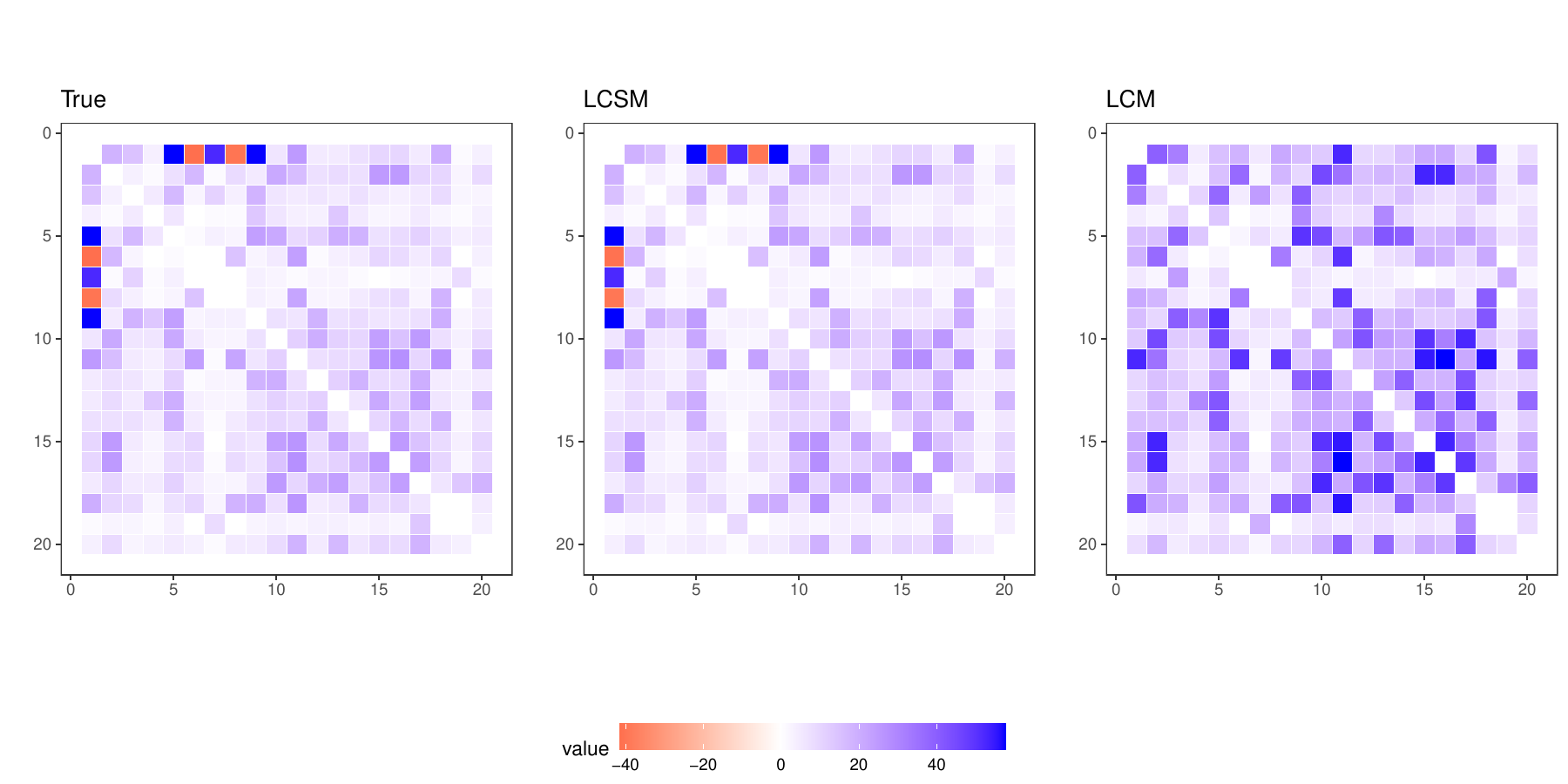}
	\end{subfigure}%
	\caption{Randomly selected estimates from LCSM (proposed method) and LCM methods. Each row associated with each type of the adjacency matrices contains three plots corresponding to the true covariance matrix and its estimates based on two methods.}
	\label{fig:sim-plot}
\end{figure}

Through $1000$ replications, we compute the scaled Frobenius error (FE) for covariance estimation and the mean squares error (MSE) for regression coefficients estimation. These quantities are defined as
\[
\text{FE}(\hat{\Sigma}) 
=
\frac{1}{\sqrt{d}}\fr{\hat{\Sigma} - \Sigma_A - \Sigma_R} 
\]
and 
\[
\text{MSE}(\hat{\theta}) 
=
\frac{1}{\card\set{j:\theta_j^* \ne 0}} 
\sum_{j:\theta_j^* \ne 0}
(\theta_j^* - \hat{\theta}_j)^2.
\]

Table \ref{tab:sim-fe} and Table \ref{tab:sim-beta} summarize the simulation results. In all considered scenarios, the proposed method (LCSM), when compared to the covariance regression estimator (LCM) studied in \cite{lan2018covariance}, demonstrates improved performance in terms of both the FE of the covariance estimator and the MSE of the coefficients estimator.
Figure \ref{fig:sim-plot} presents the estimates obtained from a randomly selected dataset for $d=20$ and $s=3$. To enhance the visibility of the covariance structure, the diagonal elements are excluded from the visualization. It is evident that the proposed penalized estimation methodology successfully detects the structure of the true covariance.
However, when using only the adjacency matrix, the covariance regression estimator based on the existing method fails to detect the structure caused by the remainder covariance. Additionally, it shows a slight overestimation of covariance components.
Simulation results indicate that the proposed method can significantly improve existing covariance regression estimators when there exists a remainder portion of the covariance that is not entirely explained by the adjacency structure.

Additionally, we measure the computation time of the proposed method to verify its scalability.
The information regarding computational details is as follows. Table \ref{tab:computing time} summarizes the average computation times and their standard errors for each simulation scenario based on $1000$ iterations. All timings were carried out on a Apple M1 MAX 3.23 GHz processor. The results in Table \ref{tab:computing time} suggest that the LCSM method is sufficiently scalable to larger datasets. Additionally, in the $1000$ iterations performed for each type, we did not encounter any cases where the positive definite correction described in Section \ref{ssec:implementation} was needed.

\begin{table}[!htb] 
	\centering 
	\resizebox{\textwidth}{!}{% 
		\begin{tabular}{l|ccc|ccc|ccc} 
			\hline 
			\textbf{Adj. Type} & \multicolumn{3}{c|}{\textbf{d = 20}} & \multicolumn{3}{c|}{\textbf{d = 50}} & \multicolumn{3}{c}{\textbf{d = 80}} \\
			& \textbf{s = 2} & \textbf{s = 3} & & \textbf{s = 2} & \textbf{s = 3} & & \textbf{s = 2} & \textbf{s = 3} & \\
			\hline 
			Type 1 & 0.127 (0.017) & 0.140 (0.019) & & 4.004 (0.210) & 4.188 (0.296) & & 27.475 (1.229) & 28.315 (1.423) & \\ 
			Type 2 & 0.114 (0.007) & 0.116 (0.011) & & 3.895 (0.207) & 4.436 (0.320) & & 26.571 (1.254) & 30.448 (1.865) & \\ 
			Type 3 & 0.123 (0.006) & 0.129 (0.008) & & 3.692 (0.210) & 3.567 (0.232) & & 22.860 (1.141) & 22.369 (1.311) & \\ 
			\hline 
		\end{tabular} 
	} 
	\caption{The average run times (seconds) and the standard errors (in parentheses) of LCSM (proposed method) with $1000$ replications simulation.}
	\label{tab:computing time} 
\end{table}

\section{Application to co-expression network analysis}\label{sec:realdata}

\subsection{Data description and modeling scheme}

%% gene co-expression network analysis
We apply our covariance regression method to estimation of co-expression network using gene-expression data obtained from the livers of mice in a specific F2 intercross. The data have been previously analyzed by \cite{Ghazalpour:2006} and used as an illustrative example in the R package \texttt{WCGNA}. For a more comprehensive understanding of the data, we refer readers to the original reference. Using the data, the existing work such as \cite{Langfelder:2007} aimed to identify clusters of genes, named eigen-genes, based on the so-called topological overlap measure (TOM) consisting of the adjacency matrix $A$ and the second-order neighborhood information $A^2$. Building upon this approach, we focus on enhancing their methodology by explicitly estimating the effect of the neighborhoods.

%% pre-processing
We begin with introducing the gene-expression data $Y = [Y_1, \ldots, Y_n]^\top$ measured over $\tilde{p}=3600$ genes from either $n_1=135$ female or $n_2=124$ male samples.
As part of data preprocessing, we remove genes with smaller coefficients of variation (CV), computed as a ratio of the sample standard deviation to the sample mean. Then, we center each gene by its empirical mean.
We adopt the pipeline of the package \texttt{WCGNA} (\cite{Langfelder:2008}) to extract the consensus network structure of two groups of samples, which plays a role of the adjacency matrix in our covariance regression. The first step of the pipeline involves computing Pearson correlation coefficients denoted by $R_\beta = (r_{k\ell})_{k,\ell}$ where $r_{k\ell} = \rho_{k\ell}^\beta$, $1 \le k,\ell \le d$. Here, we set $\beta=6$. After replacing all diagonals of $R_\beta$ by zero, we obtain the weighted adjacency matrix $A$.
The next step of \texttt{WCGNA} is to exploit the TOM to construct the similarity matrix for each group. \cite{Langfelder:2007} suggested using the quantile transformation to derive a common network structure from multiple similarity matrices. The consensus network is finally used to detect modules (gene sets/clusters), denoted as $\{\mathcal{G}_1, \ldots, \mathcal{G}_K\}$, through a clustering method (e.g. hierarchical clustering). We refer readers to the original reference (\cite{Langfelder:2007}) for more details of the procedure.

Now, we focus on a specific module, say the $k$-th, for each group of mice. We perform LCSM based on the common basis matrix defined by $A^{(k)} = ( r_{k\ell}^{(F)} \wedge r_{k\ell}^{(M)})_{k,\ell \in \mathcal{G}_k}$, where $r_{k\ell}^{(F)}$ (or $r_{k\ell}^{(M)}$) represents the power of Pearson correlation coefficient in the female (or male) group. Note that the minimum corresponds to a special case of the quantile transformation (with a level of quantile being zero).
%Thus, $A^{(k)}$ the principal submatrix of the consensus adjacency matrix on the module $\mathcal{G}_k$. The similar analysis is applied to the case of male mice.
Let $\{F_{1}^{(k)}, \ldots, F_{q}^{(k)}\}$ be the basis orthogonal to $\{I, A^{(k)}, \ldots, (A^{(k)})^s\}$. Then, we scale the basis matrices so that they are unit vectors with respect to the Frobenius norm:
\[
\tilde{A}_{j}^{(k)} = (A^{(k)})^j / ||(A^{(k)})^j||_F, \quad j=0, \ldots, s, \quad 
\tilde{F}_{\ell}^{(k)} = F_{\ell}^{(k)} / || F_{\ell}^{(k)} ||_F, \quad \ell=1, \ldots, q.
\]
We define $(A_k)^0 \equiv I$ by convention. The regression model for each module is then given by
\[
Y_i^{(k)} (Y_i^{(k)})^\top = \sum_{j=0}^s \alpha_j  \tilde{A}_{j}^{(k)} + \sum_{j=1}^q \beta_j \tilde{F}_{j}^{(k)} + \varepsilon_i \for i=1,\ldots, n,
\]
Here, $Y_i^{(k)}$ is the sub-vector restricted on the set $\mathcal{G}_k$. The similar analysis is applied to the case of male mice. We set the order $s=2$ in our analysis because we find strong multicollinearity between basis matrices of degrees higher than $3$.

{\color{black}
	Note that we can search the optimal order $s$ by model selection criteria, e.g. AIC or BIC, as we do the optimal tuning parameter $\lambda$. This approach has been used in previous work (e.g. \cite{lan2018covariance}). 
	However, it is noteworthy that increasing $s$ may introduce multicollinearity among basis matrices, as observed in the usual polynomial regression.  
}

%\PAR{Results}
\subsection{Results}%\label{ssec:implementation}
For computational simplicity, we restrict our analysis to six gene sets of size smaller than $100$, $\{\mathcal{G}_k\}_{k=1}^6$. The genes in each set are listed below.
\begin{table}[H]
	\tiny
	\begin{tabularx}{\textwidth}{|c|c|X|}
		%			\begin{tabular}{r|r|l}
		\hline
		Gene set & Size & Gene ID\\
		\hline
		1 & 37 & 00743, 05900, 06408, 08523, 08789, 09506, 12202, 12566, 12675, 13152, 24116, 25608, 26786, 30469, 35158, 36605, 38017, 46250, 46994, 49025, 52924, 53523, 54020, 54501, 55476, 58422, 59097, 60307, 63965, 63985, 67195, 67953, 68222, 69683, 71535, 74306, 78546\\
		\hline
		2 & 44 & 02021, 02758, 07836, 11900, 15218, 17456, 18479, 25048, 25601, 25842, 27378, 27861, 30620, 30812, 31586, 32920, 34467, 38270, 42929, 45252, 45911, 46218, 46735, 46768, 47216, 48535, 48659, 50031, 54705, 55391, 60559, 65821, 66276, 69238, 70677, 71411, 73157, 73365, 73735, 77517, 78706, 80165, 82041, 82551\\
		\hline
		3 & 63 & 02046, 04230, 07748, 07895, 08908, 08909, 09436, 10000, 10861, 11135, 12536, 13190, 13217, 14444, 15397, 15971, 17308, 19497, 19546, 22068, 22137, 22559, 25371, 33244, 33353, 35224, 36960, 38808, 39174, 42404, 44120, 45089, 45188, 45516, 45781, 46806, 48127, 53710, 54396, 55001, 56290, 56485, 56913, 57673, 58090, 58516, 58688, 58766, 59751, 59778, 60223, 60702, 61791, 63384, 64669, 69768, 71655, 74751, 75545, 76103, 79286, 80903, 82622\\
		\hline
		4 & 62 & 01596, 01949, 02175, 04172, 05179, 06387, 07049, 08162, 09432, 10511, 11644, 12683, 15110, 18485, 20582, 21613, 22288, 22897, 24833, 27363, 31239, 31896, 34948, 36635, 37817, 39789, 39832, 39853, 40321, 41171, 41180, 42440, 43123, 43236, 44499, 44978, 46284, 49364, 49901, 52088, 52285, 52488, 54402, 55921, 56482, 57142, 59326, 62373, 62591, 64344, 64507, 64856, 65834, 67116, 69681, 72840, 75656, 76705, 78537, 79439, 80680, 82303\\
		\hline
		5 & 82 & 00517, 00602, 01100, 03470, 04963, 05210, 05527, 06705, 07275, 07454, 09305, 10512, 11486, 12639, 14269, 15439, 16217, 21727, 23913, 30006, 30262, 30762, 30928, 31264, 33070, 33186, 33537, 38058, 40085, 40222, 40489, 41083, 42923, 43176, 43234, 46017, 46654, 47430, 48733, 50438, 51367, 51604, 51611, 52922, 53208, 53716, 54086, 54437, 54594, 56436, 57123, 57790, 58684, 59229, 60224, 60506, 60603, 62948, 65172, 65828, 66041, 67143, 67145, 67444, 67534, 68155, 69114, 70671, 73765, 73830, 74466, 75424, 76726, 77164, 78229, 78426, 78527, 78942, 80515, 80792, 81411, 82585\\
		\hline
		6 & 89 & 00231, 00525, 01394, 01646, 01947, 02151, 04888, 04929, 04959, 07970, 09297, 10988, 11143, 11559, 12141, 12375, 15549, 17950, 18332, 19091, 19519, 20073, 20368, 20566, 21626, 24043, 25869, 27298, 27700, 27858, 27978, 28983, 30412, 31280, 31742, 32606, 32698, 35101, 38501, 38919, 40465, 41568, 43889, 43985, 44003, 44175, 44489, 44497, 44766, 45019, 46069, 46375, 47183, 47198, 47239, 48203, 50551, 52236, 52548, 53547, 54025, 54545, 56241, 58285, 58531, 60201, 60736, 60741, 61427, 63950, 64112, 64537, 65211, 66426, 68097, 69653, 69884, 70955, 72657, 75851, 75907, 75932, 76328, 78106, 78212, 78286, 79610, 79636, 80864\\
		\hline
		%	\end{tabular}w
	\end{tabularx}
	\caption{List of gene IDs for each gene set. The full ID starts with MMT000$*$, i.e. MMT00000743 for the first gene ``00743'' of gene set 1.}
\end{table}

Figure \ref{fig:real_coef} presents the estimated regression coefficients. 
The proposed covariance regression model can estimate the effect sizes of $A$ and $A^2$, while \texttt{WCGNA} cannot. Moreover, we observe that the common network has different effects on the female and male groups, particularly in the gene sets 1 and 2.
\begin{figure}[H]
	\centering
	\includegraphics[width=0.6\linewidth]{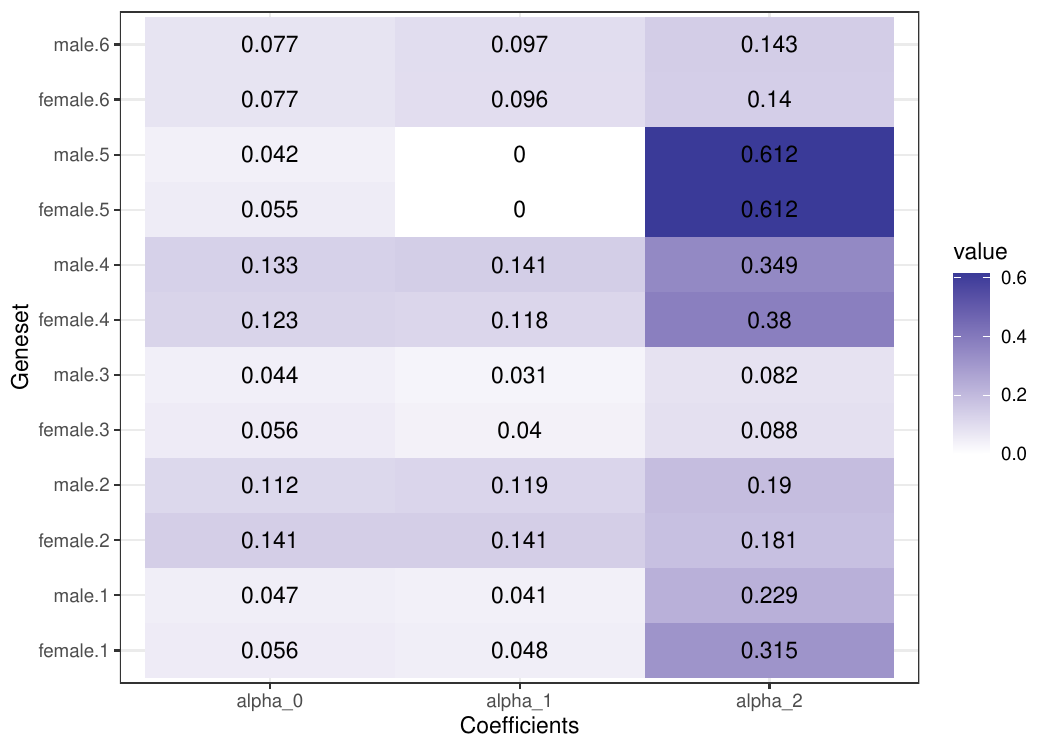}
	\caption{The estimated coefficients $(\hat{\alpha}_0, \hat{\alpha}_1, \hat{\alpha}_2)$ for each gene set of male and female groups.}
	\label{fig:real_coef}
\end{figure}

%%{\bf \cite{lan2018covariance} vs. the proposed method}
Figure \ref{fig:real_covreg} shows the heatmaps of the sample covariance matrix ($S_n$) and the covariance regression estimators ($\widehat{\Sigma}_A$,  $\widehat{\Sigma}_R$) for the first gene set of size 37. Additional results for the other sets are available in \ref{sec:add_result}. The figure confirms that the spanned basis effectively explains some portion not accounted for by the network information alone. Therefore, the proposed method demonstrates its value in recovering covariance patterns better compared to the earlier work (\cite{lan2018covariance}).

\begin{figure}[H]
	\centering
	\includegraphics[width=1\linewidth,page=3]{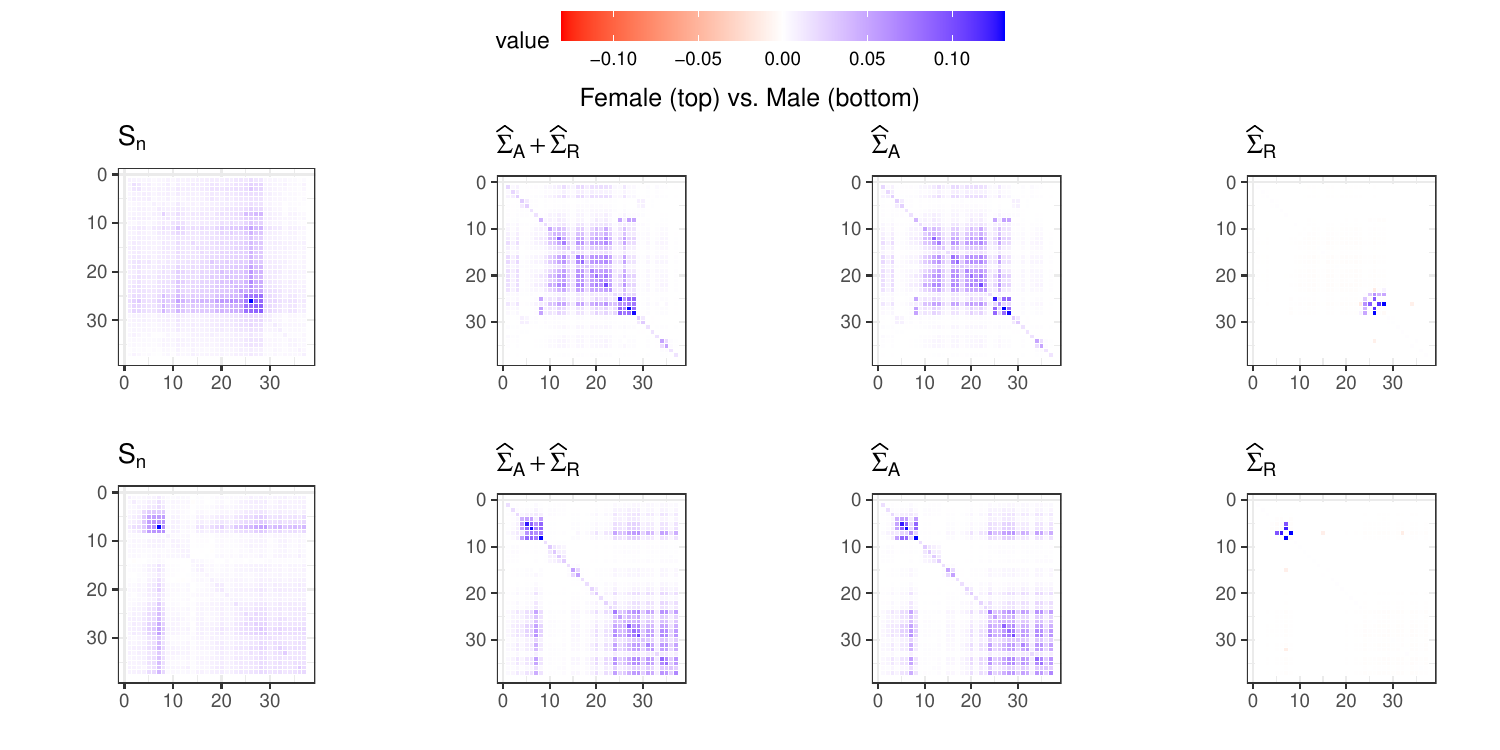}
	\caption{The empirical estimator and covariance regression estimators. The results are from the gene set 1. For each group, the rows and columns are reordered by hierarchical clustering.}
	\label{fig:real_covreg}
\end{figure}

\section{Discussion}

% % summary
In this paper, we introduce a new class of covariance regression models, called the linear covariance selection model. While existing models assume the covariance structure as a linear combination of predefined matrices, the proposed one incorporates additional basis matrices to cover covariance patterns not fully described by the given matrices. For better interpretability, we trim the fully parametrized model through the $\ell_1$-penalized regression method. 
%{\color{black}
Furthermore, we establish a statistical theory for the proposed method under the (symmetric) matrix regression framework. Though the convergence rate is sub-optimal under the Gaussian linear covariance model, we believe that this work suggests an interesting avenue for further research in symmetric matrix regression. 
%}

% % limitations and future work

It would be interesting to construct basis matrices for the remainder to incorporate the structure used in existing literature. One avenue is to consider a sparse remainder part, as in \cite{Fan:2013} wherein a covariance matrix is decomposed into low-rank plus sparse matrices. 
In our application to the liver gene expression data, we observe different effects of the common network on co-expressions between two groups. This raises a question about testing the equality of structured covariance matrices, which is equivalent to comparing regression coefficients in the context of covariance regression. Hence, new tools tailored to the inferential questions are called for. Moreover, in the co-expression analysis, the network information may not be readily available in practice but remains of primary interest. Simultaneous estimation of the network structure and its effect size could be a valuable direction for future research.

To improve estimation, one could develop estimation methods addressing the issues that $\ell_1$ penalty causes estimation bias and that covariance regression is inherently a high-dimensional problem. Techniques such as debiased lasso \citep{zhang2014confidence,van2014asymptotically}, adaptive lasso \citep{zou2006adaptive}, or non-convex penalization methods \citep{fan2001variable,fan2004nonconcave}, which are known to enjoy the oracle property in the linear regression model, can be introduced. It is expected that by incorporating these methodologies into the covariance regression model, we can derive asymptotic properties of the estimator and enhance estimation performance in high-dimensional scenarios.

\iffalse
%--------------------------------------------------------------------------------
\section*{Disclosure Statement}
%--------------------------------------------------------------------------------
No potential conflict of interest was reported by the authors.
\fi

%--------------------------------------------------------------------------------
\section*{Acknowledgements}
%--------------------------------------------------------------------------------

The research of Kwan-Young Bak was supported by the Basic Science Research Program through the National Research Foundation of Korea (NRF) funded by the Ministry of Education, Science and Technology (RS-2022-00165581).

\appendix
%--------------------------------------------------------------------------------
\section{Additional results for numerical studies}\label{sec:add_result}

We present additional results not shown in the main manuscript due to space. In the co-expression network analysis, we model the covariance matrices corresponding to the six gene sets. These results are illustrated in the following figures. We can see that the remainder matrix ($\hat{\Sigma}_R$) ameliorates the approximation of the unstructured estimators ($S_n$).
\begin{figure}[H]
	\centering
	\includegraphics[width=1\linewidth,page=3]{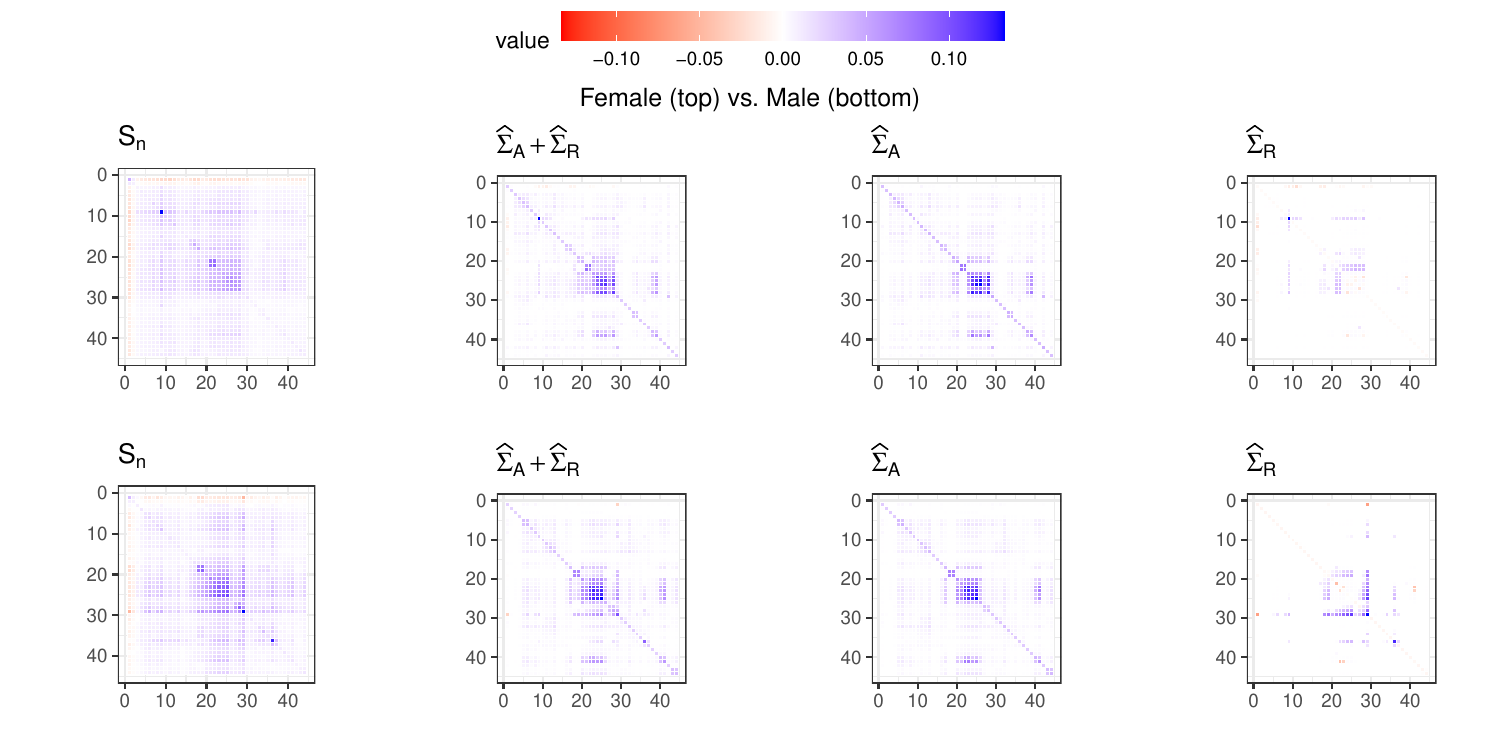}
	\includegraphics[width=1\linewidth,page=3]{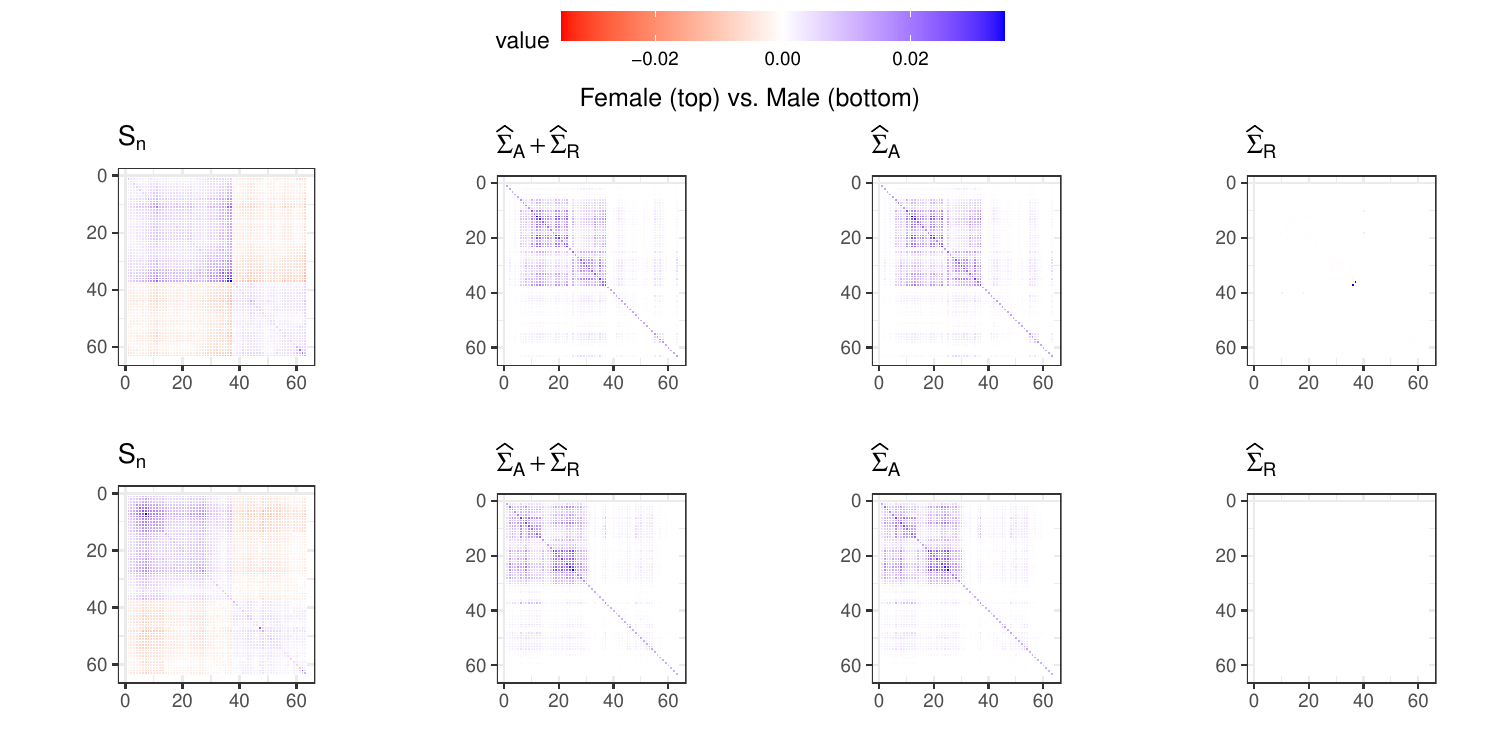}
	\caption{The empirical estimator and covariance regression estimators. The results are from the gene set 2 (top) and 3 (bottom). For each group, the rows and columns are reordered by hierarchical clustering.}
	%	\label{fig:real_covreg}
\end{figure}

\begin{figure}[H]
	\centering
	\includegraphics[width=1\linewidth,page=3]{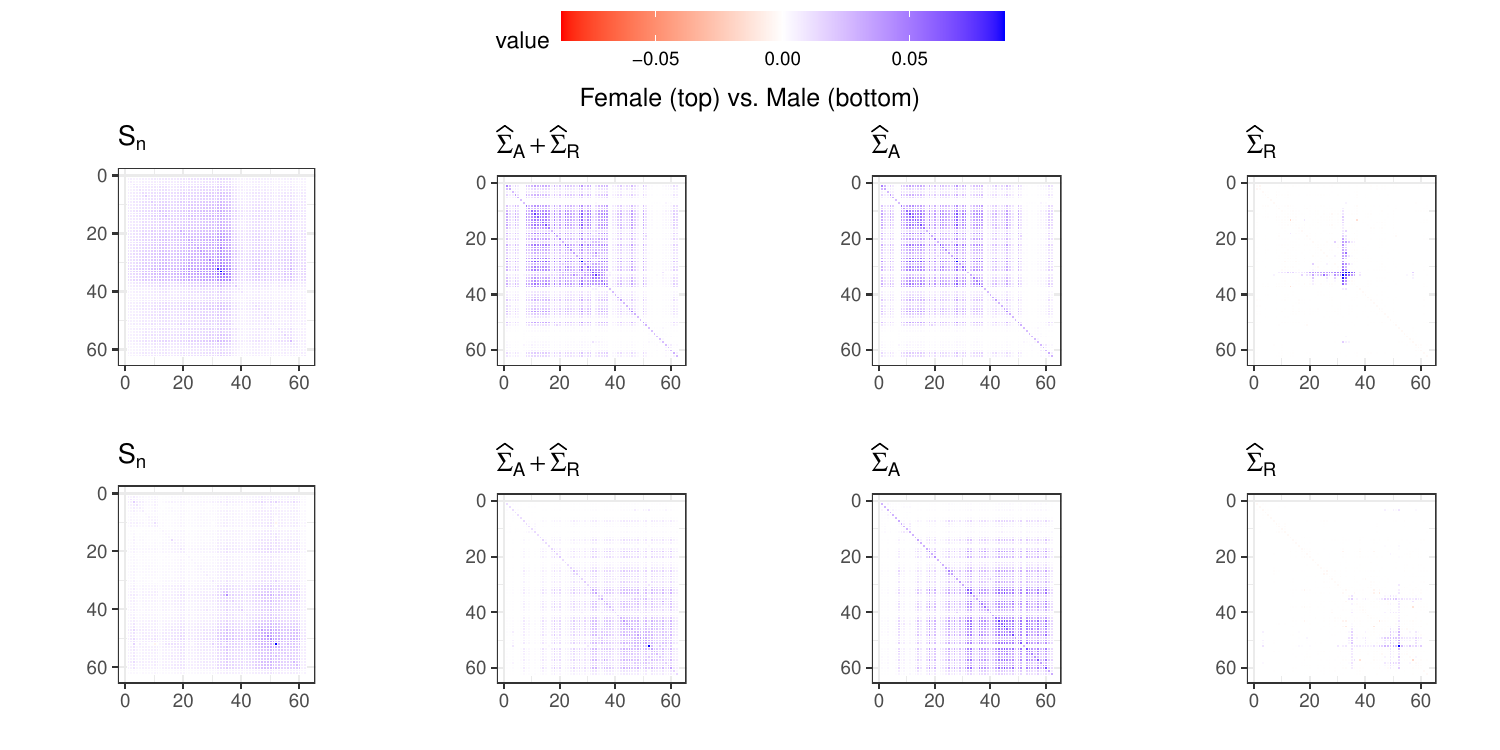}
	\includegraphics[width=1\linewidth,page=3]{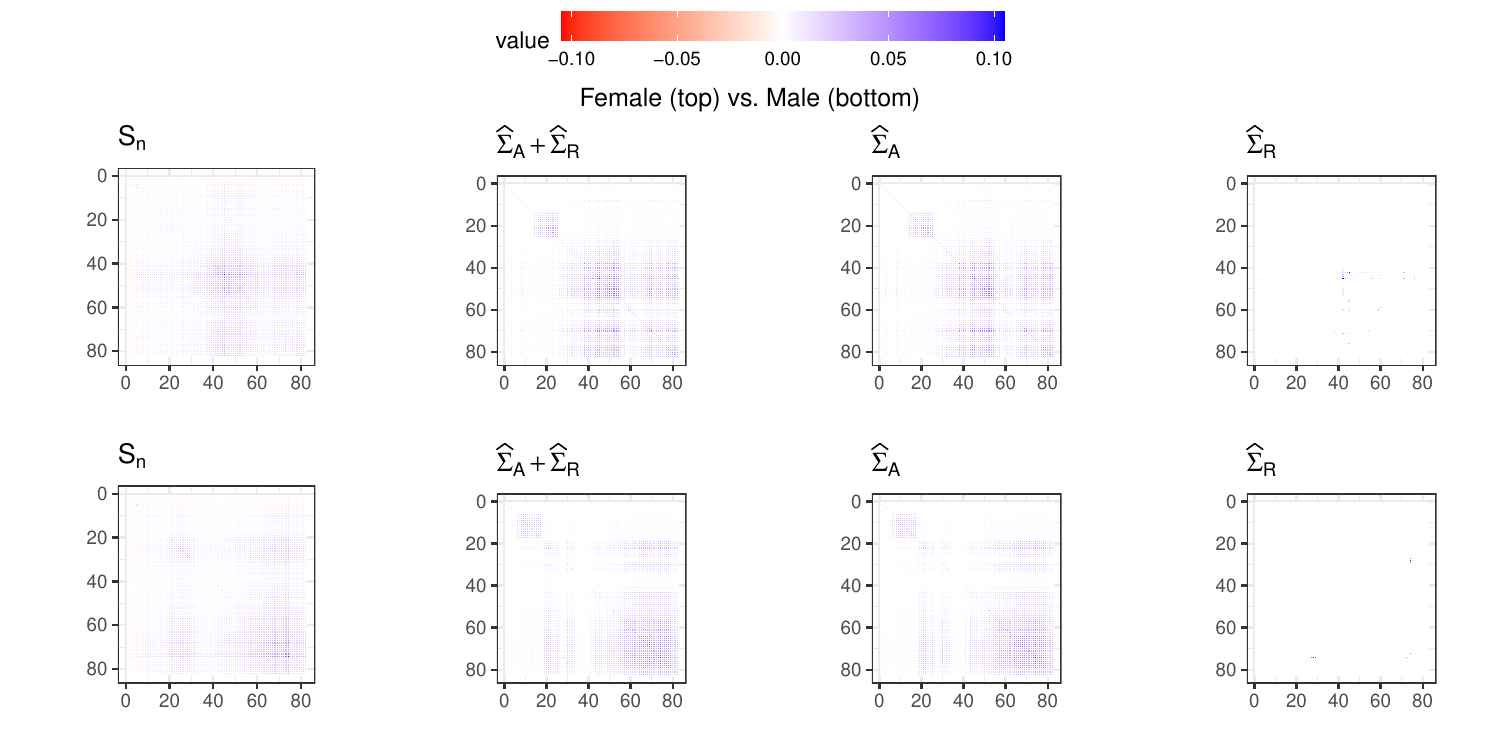}
	\caption{The empirical estimator and covariance regression estimators. The results are from the gene set 4 (top) and 5 (bottom). For each group, the rows and columns are reordered by hierarchical clustering.}
	%	\label{fig:real_covreg}
\end{figure}

\begin{figure}[H]
	\centering
	\includegraphics[width=1\linewidth,page=3]{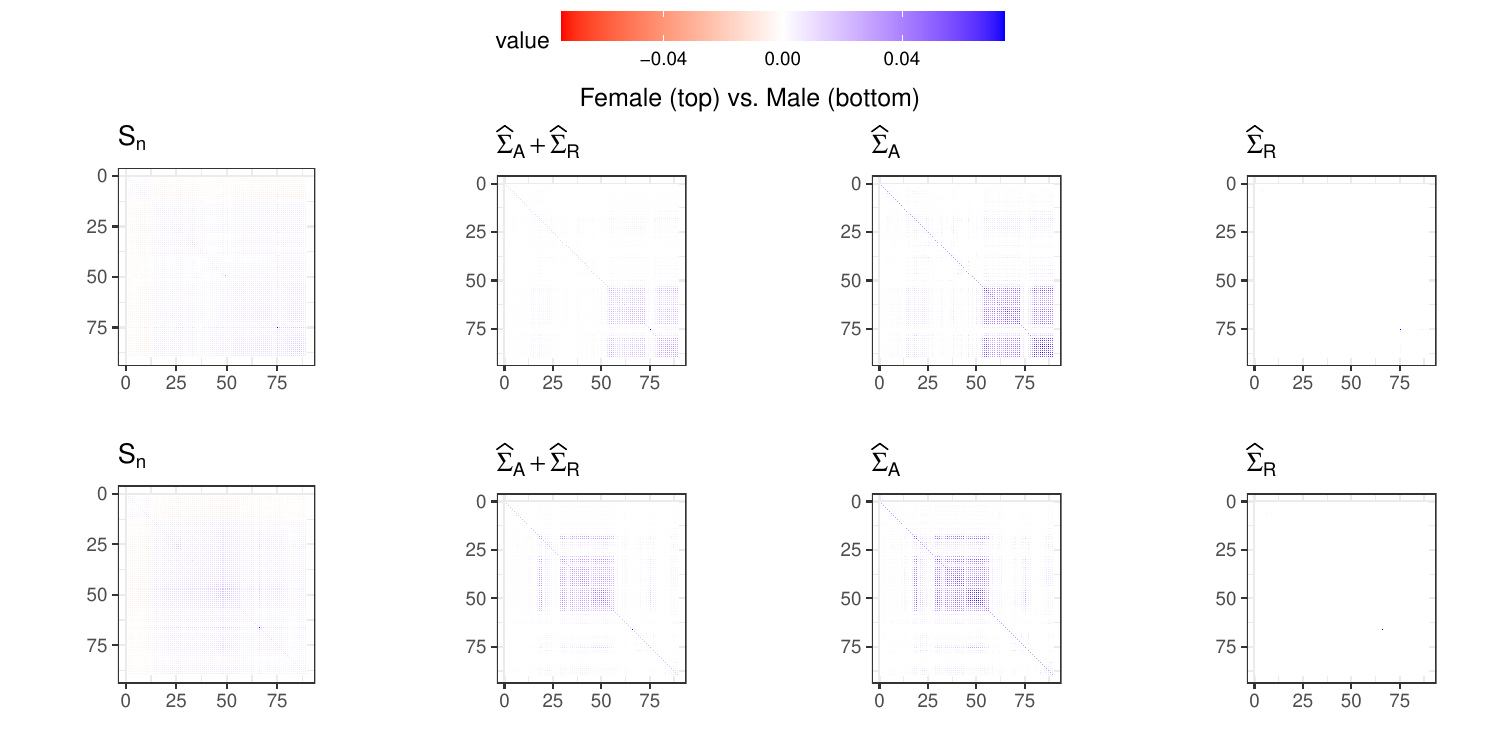}
	\caption{The empirical estimator and covariance regression estimators. The results are from the gene set 6. For each group, the rows and columns are reordered by hierarchical clustering.}
	%	\label{fig:real_covreg}
\end{figure}

\section{Proofs}\label{sec:proofs}

\subsection{Proof of Proposition \ref{prop:update-formula}}
To obtain the closed form solution of the univariate optimization problem, we make use of the following result: 
\begin{equation}\label{eq:univ-ST}
	\argmin_{z \in \R} (z^2 - 2az + 2\lambda \av{z}) = \ST(a,\lambda).
\end{equation}
This can be easily verified by separately examining the axis of symmetry of a parabola for cases in which $a>\lambda>0$, $a<-\lambda<0$ and $\av{a} < \lambda$.
Note
\begin{align*}
	\ell_j^\lambda(\theta_j) 
	&= 
	\sum_{i=1}^n \sum_{l,k=1}^d  \br{Z_{i, kl} - (\tilde{\theta}_1 B_{1,kl} + \cdots + \tilde{\theta}_{j-1} B_{j-1, kl} + \theta_j B_{j,kl} + \tilde{\theta}_{j+1} B_{j+1, kl}  + \cdots + \tilde{\theta}_p B_{p, kl})}^2 
	\\
	&\quad + 
	2\lambda \av{\theta_j} + 2\lambda \sum_{j' \ne j} \av{\tilde{\theta}_{j'}}
	\\
	&=
	\sum_{i=1}^n \sum_{k,l=1}^d \br{r_{ijkl} - \theta_j B_{j,kl}}^2 + 2\lambda \av{\theta_j} + C
	\for 
	r_{ijkl} = Z_{i,kl} - \sum_{j' \ne j} \tilde{\theta}_{j'} B_{j',kl}
	\\
	&=
	\theta_j^2 \sum_{i=1}^n \sum_{k,l=1}^d B_{j,kl}^2
	-
	2 \theta_j \sum_{i=1}^n \sum_{k,l=1}^d r_{ijkl} B_{j,kl} + 2\lambda \av{\theta_j} + C
	\\
	&=
	n \sum_{k,l=1}^d B_{j,kl}^2
	\pa{
		\theta_j^2 - 2 \frac{\sum_{i=1}^n \sum_{k,l=1}^p r_{ijkl} B_{j,kl}}{n \sum_{k,l=1}^d B_{j,kl}^2}
		+
		2 \frac{\lambda}{n \sum_{k,l=1}^d B_{j,kl}^2} \av{\theta_j} 
	}
	+ C,
\end{align*}
where $C$ is the term independent of $\theta_j$. Using \eqref{eq:univ-ST}, the update formula based on the minimizer of $\ell_j^\lambda$ is given by
\[
\tilde{\theta}_j \leftarrow \ST \pa{\frac{\sum_{i=1}^n \sum_{k,l=1}^d r_{ijkl} B_{j,kl}}{n \sum_{k,l=1}^d B_{j,kl}^2},  \frac{\lambda}{n \sum_{k,l=1}^d B_{j,kl}^2}}.
\]

\subsection{Proof of Proposition \ref{prop:lambda-max}}

The vectorization of a matrix $D \in \R^{p \times p}$ is defined as
\[
\v(D) = 
\BM
\c_1(D) \\
\vdots \\
\c_p(D) 
\EM
\in \R^{p^2},
\]
where $\c_j(D)$ denotes the $j$-th column of $D$. 
Let 
\[
\Z = 
\BM
\v(Z_1) \\
\vdots \\
\v(Z_n) \\
\EM
\in \R^r
,
\quad 
\Omega =
\BM
\v(\Sigma) \\
\vdots \\
\v(\Sigma)
\EM
\in 
\R^r
,
\quad
\X = 
\BM
\v(B_1) & \cdots & \v(B_p) \\
\vdots & \ddots & \vdots \\
\v(B_1) & \cdots & \v(B_p)
\EM
\in \R^{r \times p},
\]
where $r = nd^2$. Observe that $\Omega = \X \theta^*$.
With this notation, we have
\[
\ell^\lambda(\theta) = \av{\Z - \X \theta}^2 + 2\lambda \av{\theta}_1,
\]
where $\av{\cdot}$ is the $\ell_2$ norm of a Euclidean vector. Thus, the subgradient optimality condition \citep{rockafellar1997convex} is given by
\[
0 = - \X^\top (\Z - \X \theta) + \lambda \sign(\theta),
\]
where $\sign(\theta) = (\sign(\theta_j)) \in \R^p$ is defined as
\[
\sign(\theta_j) \in
\begin{cases}
\set{+1} & \mbox{if $\theta_j > 0$} \\
\set{-1} & \mbox{if $\theta_j < 0$} \\
[-1,+1] & \mbox{if $\theta_j = 0$.} \\
\end{cases}
\]
Observe
\[
\lambda_m 
= 
\max_{1 \le j \le p} \Av{\sum_{i=1}^n \ip{B_j}{Z_i}}
=
\max_{1 \le j \le p} \Av{\c_j(\X)^\top \Z}.
\]
Choose $\lambda > \lambda_m$ and suppose that there exists a non-zero vector $\hat{\theta}^\lambda$ that satisfies the optimality condition
\[
\X^\top \X \hat{\theta}^\lambda = \X^\top \Z - \lambda \sign(\hat{\theta}^\lambda).
\]
The column independence of $\X$ implies that 
\[
(\hat{\theta}^\lambda)^\top (\X^\top \Z - \lambda \sign(\hat{\theta}^\lambda) ) = (\hat{\theta}^\lambda)^\top\X^\top \X \hat{\theta}^\lambda > 0.
\]
On the other hand, we find 
\[
(\hat{\theta}^\lambda)^\top (\X^\top \Z - \lambda \sign(\hat{\theta}^\lambda) )
=
\sum_{j=1}^p \hat{\theta}^\lambda_j \pa{\c_j(\X)^\top \Z - \lambda \sign (\hat{\theta}^\lambda_j)}.
\]
By the definition of $\lambda_m$, $\hat{\theta}^\lambda_j > 0$ implies $\c_j(\X)^\top \Z - \lambda \sign (\hat{\theta}^\lambda_j) \le 0$ and $\hat{\theta}^\lambda_j < 0$ implies $\c_j(\X)^\top \Z - \lambda \sign (\hat{\theta}^\lambda_j) \ge 0$ so that
\[
(\hat{\theta}^\lambda)^\top (\X^\top \Z - \lambda \sign(\hat{\theta}^\lambda) ) \le 0,
\]
which leads to a contradiction. Therefore, we have $\hat{\theta}^\lambda = 0$ for $\lambda > \lambda_m$.

\subsection{Proofs of theoretical results}

%\color{black}
\subsubsection*{Verification of Bernstein moment condition}\label{ssec:verify_bernstein}

We state the result from \cite{zhu2012short} for convenience of readers, which verifies the Bernstein condition from the multivariate Gaussian distribution.
\begin{proposition}[Lemma 4 of \cite{zhu2012short}]
	For $Y \sim N_d(0, \Sigma)$, we have
	\[
	\text{\rm var}(YY^\top - \Sigma) = \tr(\Sigma) \Sigma + \Sigma^2 \text{ and } \mathbb{E}(YY^\top - \Sigma)^k \preceq \dfrac{1}{2}k! \tr(\Sigma)^{k-2} \{\tr(\Sigma) \Sigma + \Sigma^2\}.
	\]
\end{proposition}

The following shows a sub-Gaussian vector satisfies another moment condition, similar to the Bernstein condition, which suffices to obtain the concentration inequality.
\begin{lemma}\label{le:subexp_YYt}
	% % assumptions
	Suppose $Y$ is a sub-Gaussian vector, i.e. its projection onto any unit vector is sub-Gaussian in $\R$. Moreover, we assume $\mathbb{E}[Y]=0$ and $\mathbb{E}[YY^\top] = 
	\Sigma$
	% % claim
	Then, a random matrix $Q \equiv YY^\top - \Sigma$ satisfies for some $B>0$,
	\begin{equation}\label{eq:moment_ineq}
		\mathbb{E}[Q^k] \preceq \dfrac{1}{2} k! B^k I, \quad k\ge 2.
	\end{equation}
	Remark that the constant $B$ may depend on the dimension $p$ and $\Sigma$.
	
\end{lemma}

\begin{proof}
	To prove \eqref{eq:moment_ineq}, it is enough to show $||\mathbb{E}[Q^k] ||_2 \le 0.5  k! B^k$. 
	For a symmetric matrix $A$, we have
	\[
	||\mathbb{E}[A] ||_2 = \sqrt{\sup_{t: ||t||_2=1} t^\top \mathbb{E}[A] \mathbb{E}[A] t} \le \sqrt{\sup_{t: ||t||_2=1} t^\top \mathbb{E}[A^2] t} \le
	\sqrt{\mathbb{E} \left[\sup_{t: ||t||_2=1} t^\top [A^2] t\right] } = \sqrt{ \mathbb{E} ||A||_2^2}
	\]
	Also, we note that by the triangular inequality 
	\[
	\norm{Q^k}_2 \le \norm{Q}_2^k \le (\norm{YY^\top}_2 + \norm{\Sigma}_{2})^k =(Y^\top Y + \norm{\Sigma}_{2})^k.
	\]
	where the last inequality is from $\norm{YY^\top}_2 = Y^\top Y$ since $YY^\top$ is a rank-one matrix. By using $(a + b)^k \le 2^{k-1}(a^k + b^k)$ for $a,b>0$ and $k \ge 1$, we get
	\[
	(Y^\top Y + \norm{\Sigma}_{2})^{2k} \le  2^{2k-1} ((Y^\top Y)^{2k} + \norm{\Sigma}_{2}^{2k}).
	\]
	Combining these, we get 
	\[
	\norm{\mathbb{E}[Q^k]}_2 \le \sqrt{2^{2k-1}(\mathbb{E}(Y^\top Y)^{2k} + \norm{\Sigma}_{2}^{2k})}
	\]
	Because $Y^\top Y=\sum_{j=1}^p Y_j^2$ is the sum of square of sub-Gaussian variables, it is sub-exponential, implying its moments satisfy $\mathbb{E}(Y^\top Y)^k \le (B')^k k!$ for $k\ge 2$ for some $B'>0$. Then, we get
	\[
	\norm{\mathbb{E}[Q^k]}_2 \le \sqrt{2^{2k-1}((B')^{2k} (2k)! + \norm{\Sigma}_{2}^{2k})} \le 
	\sqrt{2^{2k-1}((B')^{2k} (k! 2^k)^2 + \norm{\Sigma}_{2}^{2k})}
	\]
	by using $(2k)! \le (k! 2^k)^2$. As a consequence, a simple calculation with $B = 4(B' \vee \norm{\Sigma}_{2})$ leads to 
	\[
	\norm{\mathbb{E}[Q^k]}_2 \le 0.5 k! B^k.
	\]
	%	
	%	
	%	where $Z = \Sigma^{-1/2}Y$ is an isotropic sub-Gaussian random vector. 
	%	
	%	, which implies that $||ZZ^\top - I||_2 = \max\{|Z^\top Z-1|, 1\}$. Combining these, we have
	%	\[
	%	\norm{\mathbb{E}[Q^k]}_2 \le \sqrt{ \mathbb{E} \norm{Q^k}_2^2} \le \norm{\Sigma}_2^{k} \sqrt{\mathbb{E} \big[\max\{|Z^\top Z-1|^{2k}, 1\} \big]}.
	%	\]
	%	By using $(a + b)^k \le 2^{k-1}(a^k + b^k)$ for $a,b>0$ and $k \ge 1$, we get
	%	\[
	%	\max\{|Z^\top Z-1|^{2k}, 1\} \le
	%	 2^{2k-1}(|Z^\top Z|^{2k} + 1), \quad k\ge 1.
	%%	\max\{|Z^\top Z-p+ (p-1)|^{2k}, 1\} = 
	%%	\le 
	%%	 2^{2k-1}(|Z^\top Z-1|^{2k} + 1)
	%%	 \le 2^{2k-1}(2^{2k-1} (|Z^\top Z-p|^{2k} + (p-1)^{2k-1}) + 1)
	%	\]
	%	
	%	
	%	
	%	Combining these, we get $\norm{\mathbb{E}[Q^k]}_2 \le \sqrt{2^{2k-1}(\mathbb{E}(Z^\top Z)^{2k} + 1)}$. Since $Y^\top Y=\sum_{j=1}^p Y_j^2$ is a sub-exponential random variable, its moments satisfy $\mathbb{E}(Z^\top Z)^k \le (B')^k k!$ for $k\ge 2$ for some $B'>0$. Hence, 
	%	\[
	%	\norm{\mathbb{E}[Q^k]}_2 \le 
	%	 \sqrt{2^{2k} B^{2k} (2k)!} \le 2^{k} B^{k} 2^k k!,
	%	\]
	%	where $B = B' \vee 1$.
	%%	\dfrac{1}{2}
\end{proof}

\color{black}

\subsubsection*{Existence and uniqueness}
\begin{lemma}\label{le:exist-unique}
	Suppose that the basis matrices $\set{B_j}_{j=1}^p$ are linearly independent. Then, the minimizer of $\ell^{\lambda}(\cdot)$ uniquely exists.
\end{lemma}

\begin{proof}

	(Existence) The minimization problem for $\ell^\lambda(\cdot)$ can be re-expressed as 
	\[
	\text{minimize }\ \ell(\theta) \quad  \text{subj. to. }\ \av{\theta}_1 \le C,
	\]
	where a positive constant $C$ is in one-to-one relation to the complexity parameter $\lambda$ via the Lagrangian duality. Since the empirical risk function $\ell(\cdot)$ is continuous on a compact set $\set{\theta \in \R^p \colon \av{\theta}_1 \le C}$, the existence of an optimal solution $\hat{\theta}$ is guaranteed.

	(Uniqueness) 
	Let $\hat{\theta}^{(1)}$ and $\hat{\theta}^{(2)}$ be distinct minimizers of $\ell^\lambda(\cdot)$. Recall that the minimum of a convex function is unique and the minimizers form a convex set. Hence, the convex combination $\alpha \hat{\theta}^{(1)} + (1-\alpha)\hat{\theta}^{(2)}$ ($\alpha \in (0,1)$) achieves the same minimum value denoted by $m^*=\ell^\lambda(\hat{\theta}^{(1)})=\ell^\lambda(\hat{\theta}^{(2)})$. Using the strict convexity of $\theta \mapsto |\Z - \X \theta|^2 + \lambda |\theta|_1$, we have
	\[
	\begin{array}{l}
	m^*=|\Z - \X (\alpha \hat{\theta}^{(1)} + (1-\alpha)\hat{\theta}^{(2)})|^2 + 2\lambda |\alpha \hat{\theta}^{(1)} + (1-\alpha)\hat{\theta}^{(2)}|_1 \\
	\qquad\qquad < 
	\alpha(|\Z - \X \hat{\theta}^{(1)}|^2 + 2\lambda  |\hat{\theta}^{(1)}|_1) + 
	(1-\alpha)(|\Z - \X \hat{\theta}^{(2)})|^2 + 2\lambda |\hat{\theta}^{(2)}|_1) = m^*,
	\end{array}
	\]
	which leads to a contradiction. As a result, the minimizer is unique.

	%
	%Let $\hat{\theta}^{(1)}$ and $\hat{\theta}^{(2)}$ be minimizers of $\ell^\lambda(\cdot)$, and consider the convex combination $\hat{\theta}^{(2)} + \kappa \eta$ for $\eta = \hat{\theta}^{(1)} - \hat{\theta}^{(2)}$ with $\kappa \in (0,1)$. Letting $q(\kappa) = \ell^\lambda(\hat{\theta}^{(2)} + \kappa \eta)$, we have
	%\[
	%\frac{\kappa}{\partial \kappa} q(\kappa) 
	%=
	%\frac{\kappa}{\partial \kappa} \ell(\hat{\theta}^{(2)} + \kappa \eta) + \lambda \sum_{j=1}^p \sign(\hat{\theta}_j^{(2)} + \kappa \eta).
	%\]
	%Recall that the minimum of a convex function is unique and the minimizers form a convex set so that $q(\kappa)$ equals to a constant on a subinterval of $(0,1)$. In addition, the continuity of  $\kappa \mapsto \hat{\theta}^{(2)} + \kappa \eta$ implies that $\sign(\hat{\theta}_j^{(2)} + \kappa \eta)$ is a constant for all $j$ on an open subinterval of $(0,1)$. It follows that  $\ell(\hat{\theta}^{(2)} + \kappa \eta)$ is constant on that interval. Note that the value of $\ell(\hat{\theta}^{(2)} + \kappa \eta)$ depends on $\kappa$ only through the term $\kappa \sum_{j=1}^p \eta_j B_j$. Since this term must be a constant for all $\kappa$ in the subinterval, we find that $ \sum_{j=1}^p \eta_j B_j = 0$. The linear independence of $\set{B_j}_{j=1}^p$ implies that $\eta_j = 0$ for all $j$ so that $\hat{\theta}^{(1)} = \hat{\theta}^{(2)}$, which establishes the uniqueness of the estimator.
\end{proof}
\resetconstant

\subsubsection*{Technical lemmas}

In the sequel, we adopt the notations used to denote the vectorized quantities in the proof of Proposition \ref{prop:lambda-max}.

\begin{lemma}\label{le:ip}
	For symmetric matrices $B$ and $D$ in $\R^{d \times d}$, 
	\[
	\av{\ip{B}{D}}
	\le 
	\norm{B}_* \norm{D}_2 
	\le 
	\sqrt{\rank(B)} \fr{B} \norm{D}_2.
	\]
	where $\norm{\cdot}_*$ is the nuclear norm of a matrix.
\end{lemma}
\begin{proof}
	Let $e_j$ denote the vector with $1$ in the $j$-th element and $0$ elsewhere. Observe, for an arbitrary matrix $E \in \R^{d \times d}$ and $j=1,\ldots, d$,
	\[
	\norm{E}_2
	= 
	\sup_{x \ne 0} \frac{\av{Ex}}{\av{x}} 
	\ge 
	\frac{\av{Ee_j}}{\av{e_j}} 
	=
	\av{\c_j(E)}
	\ge \av{E_{lj}} \for l=1,\ldots, d
	\]
	so that $\norm{E}_2 \ge \max_{1 \le k,j \le d} \av{E_{kj}}$. Consider the singular value decomposition of $B=PMP^\top$, where $P^\top P = PP^\top = I$ and $M$ consists of the descending order singular values $\mu_1(B),\ldots, \mu_d(B)$. 
	Observe 
	\begin{align*}
		\ip{B}{D}
		&=
		\tr(B^\top D)
		=
		\tr(PMP^\top D)
		\\
		&=
		\tr(MP^\top D P)
		=
		\sum_{j=1}^d (P^\top D P)_{jj} \mu_j(B).
	\end{align*}
	Using the upper bound property of the spectral norm proven above and the orthogonal invariance of the spectral norm, we have
	\begin{align*}
		\av{\ip{B}{D}}
		&\le 
		\norm{P^\top D P}_2 \sum_{j=1}^d \mu_j(B)
		=
		\norm{D}_2\sum_{j=1}^{\rank(B)} \mu_j(B)
		=
		\norm{D}_2 \norm{B}_*.
	\end{align*}
	In addition, the Cauchy-Schwarz inequality implies that
	\begin{align*}
		\av{\ip{B}{D}}
		&\le 
		\norm{D}_2\sum_{j=1}^{\rank(B)} \mu_j(B)
		\le 
		\norm{D}_2\sqrt{\sum_{j=1}^{\rank(B)} 1} \sqrt{\sum_{j=1}^{\rank(B)}  \mu_j^2(B)}
		=
		\sqrt{\rank(B)} \fr{B} \norm{D}_2.
	\end{align*}
\end{proof}

\begin{lemma}\label{le:ip-bd}
	Suppose Assumption \ref{cond:B2-bd} holds.
	For $\varepsilon = (\v(\varepsilon_1)^\top, \ldots, \v(\varepsilon_n)^\top)^\top \in \R^{r}$, we have
	\[
	(\X^\top \varepsilon)^\top (\hat{\theta} - \theta^*)
	\le 
	u_p \Mr{B2-bd} \Norm{\sum_{i=1}^n \varepsilon_i}_2 \av{\hat{\theta} - \theta^*}_1.
	\]
\end{lemma}

\begin{proof}
	It follows from Lemma \ref{le:ip} and Assumption \ref{cond:B2-bd} that
	\begin{align*}
		(\X^\top \varepsilon)^\top (\hat{\theta} - \theta^*)
		&\le 
		\sum_{j=1}^p \av{\c_j(\X)^\top \varepsilon} \av{\hat{\theta}_j - \theta_j^*}
		=
		\sum_{j=1}^p \Av{\sum_{i=1}^n \v(B_j)^\top \v(\varepsilon_i)} \av{\hat{\theta}_j - \theta_j^*}
		\\
		&=
		\sum_{j=1}^p \Av{\sum_{i=1}^n \v(B_j)^\top \v(\varepsilon_i)} \av{\hat{\theta}_j - \theta_j^*}
		=
		\sum_{j=1}^p \Av{ \ip{B_j}{\sum_{i=1}^n \varepsilon_i}} \av{\hat{\theta}_j - \theta_j^*}
		\\
		&\le 
		\Norm{\sum_{i=1}^n \varepsilon_i}_2 \sum_{j=1}^p \sqrt{\rank(B_j)}\norm{B_j}_2 \av{\hat{\theta}_j - \theta^*_j}
		\le 
		u_p \Mr{B2-bd} \Norm{\sum_{i=1}^n \varepsilon_i}_2 \av{\hat{\theta} - \theta^*}_1.
	\end{align*}
\end{proof}

\begin{lemma}\label{le:risk-ub}
	Suppose Assumption \ref{cond:B2-bd} holds. Then, we have
	\[
	\fr{\hat{\Sigma} - \Sigma }^2 
	\le 
	\frac{2 u_p \Mr{B2-bd}}{n} \Norm{\sum_{i=1}^n \varepsilon_i}_2 \av{\theta^* - \hat{\theta}}_1
	+
	\frac{2\lambda}{n} (\av{\theta^*}_1 - \av{\hat{\theta}}_1).
	\]
\end{lemma}

\begin{proof}
	Due to the optimality of $\hat{\theta}$, we have $\ell^\lambda(\hat{\theta})\le \ell^\lambda(\theta^*)$, and some basic algebra would result in the following inequality
	\begin{align*}
		n\fr{\hat{\Sigma} - \Sigma}^2
		&= 
		\av{\X \hat{\theta} - \X \theta^*}^2
		\\
		&\le 
		2 (\X^\top \varepsilon)^\top (\hat{\theta} - \theta^*)
		+2\lambda (\av{\theta^*}_1 - \av{\hat{\theta}}_1)
		\\
		&\le 
		2 u_p \Mr{B2-bd} \Norm{\sum_{i=1}^n \varepsilon_i}_2 \av{\theta^* - \hat{\theta}}_1
		+2\lambda (\av{\theta^*}_1 - \av{\hat{\theta}}_1).
	\end{align*}
	We applied Lemma \ref{le:ip-bd} in the last inequality. Dividing both sides by $n$ yields the desired result.
\end{proof}

\subsubsection*{Large deviation}

Lemma \ref{le:hoef-ineq} and Lemma \ref{le:bernstein-ineq} present matrix generalization of the Hoeffding and Bernstein inequalities. Proofs can be found in, for example, \cite{wainwright2019high}.

\begin{lemma}\label{le:hoef-ineq}
	Consider a sequence of independent, zero-mean, symmetric random matrices $\set{X_i}_{i=1}^n$ in $\R^{d \times d}$ that satisfy the sub-Gaussian condition with parameters $\set{W_i}_{i=1}^n$. Then, for all $t > 0$, we have
	\[
	\pr\pa{\frac{1}{n} \Norm{\sum_{i=1}^n X_i}_2 > t}
	\le 
	2 \rank\pa{\sum_{i=1}^n V_i} \exp\pa{-\frac{nt^2}{2\sigma_{W,n}^2}}.
	\]
\end{lemma}

\begin{lemma}\label{le:bernstein-ineq}
	Consider a sequence of independent, zero-mean, symmetric random matrices $\set{X_i}_{i=1}^n$ in $\R^{d \times d}$ that satisfy the Bernstein moment condition with parameter $b > 0$. Then, for all $t > 0$, we have
	\[
	\pr\pa{\frac{1}{n} \Norm{\sum_{i=1}^n X_i}_2 > t}
	\le 
	2 \rank\pa{\sum_{i=1}^n \Var(X_i)} \exp\pa{-\frac{nt^2}{2(\sigma_{X,n}^2 + bt)}},
	\]
	where $\sigma_{X,n}^2 = \frac{1}{n} \norm{\sum_{i=1}^n \Var(X_i)}_2$.
\end{lemma}

A large deviation bound similar to Lemma \ref{le:bernstein-ineq} can be derived under the assumption \eqref{eq:moment_ineq}. This result guarantees that the results of Theorem \ref{th:nonasymptotic-bd} and Remark \ref{rem:asymp-rate} for sub-exponential case hold when $Y_i$, $i=1,\ldots, n$, follows a sub-Gaussian distribution.
\begin{lemma}\label{le:bernstein-ineq-subGaussian}
	Suppose that a sequence of independent, zero-mean, symmetric random matrices $\set{X_i}_{i=1}^n$ in $\R^{d \times d}$ satisfy the moment condition \eqref{eq:moment_ineq}. Then, for all $t > 0$, we have
	\[
	\pr\pa{\frac{1}{n} \Norm{\sum_{i=1}^n X_i}_2 > t}
	\le 
	2 d \exp\pa{-\frac{nt^2}{2(1 + bt)}}.
	\]
\end{lemma}
\begin{proof}
	Under \eqref{eq:moment_ineq}, we have 
	\[
	\E[X_i^k] \preceq  \frac{1}{2} k! b^{k} I \for k =2,3,\ldots \quad \text{and} \quad i = 1,\ldots, n.
	\]
	It follows that, for $\av{s} \le  1/b$, 
	\begin{align*}
		m_{X_i}(s) 
		&=
		I + \sum_{j=2}^\infty \frac{s^j \E[X_i^j]}{j!}
		\preceq 
		I + \frac{b^2s^2 I}{2(1-b\av{s})}
		\\
		&\preceq 
		\exp\pa{\frac{b^2s^2 I}{2(1-b\av{s}}}.
	\end{align*}
	This with Lemma 6.13 of \cite{wainwright2019high} implies that 
	\[
	\tr\pa{\exp\pa{\sum_{i=1}^{n} \log m_{X_i}(s)}}
	\le 
	\tr\pa{\exp\pa{\frac{b^2s^2 I}{2(1-b\av{s}}}}
	\le 
	d \exp\pa{\frac{b^2 s^2}{1-b\av{s}}},
	\]
	for $i=1,\ldots, n$.
	The Chernoff's bound implies, for $t>0$,  
	\[
	\pr\pa{\frac{1}{n} \Norm{\sum_{i=1}^n X_i}_2 > t}
	\le 
	2d \exp\pa{\frac{nb^2 s^2}{1-b\av{s}} - nts},
	\]
	and the desired result follows for $s = t/(1 + bt)$.
\end{proof}

\begin{lemma}\label{le:large-dev-gauss}
	Suppose Assumption \ref{cond:sub-Gaussian} holds.
	Define an event 
	\[
	\mc{A} =  \set{\frac{1}{n} \Norm{\sum_{i=1}^n \varepsilon_i}_2 \le \frac{\lambda}{\Mr{B2-bd} u_p n} }.
	\]
	Choosing $\lambda$ as \eqref{eq:lambda-gauss}, we obtain
	\[
	\pr(\mc{A}^c) \le \nu.
	\]
\end{lemma}

\begin{proof}
	Lemma \ref{le:hoef-ineq} implies
	\begin{align*}
		\pr(\mc{A}^c) 
		&= 
		\pr\pa{\set{\frac{1}{n} \Norm{\sum_{i=1}^n \varepsilon_i}_2 > \frac{\lambda}{\Mr{B2-bd} u_p n} }}
		\le 
		\pr\pa{\frac{1}{n} \Norm{\sum_{i=1}^n \varepsilon_i}_2 > \frac{\lambda}{\Mr{B2-bd} u_p n}}
		\\
		&\le 2d \exp\pa{-n \frac{\lambda^2 / (n^2\Mr{B2-bd}^2 (u_p)^2)}{2\sigma_{W,n}^2}}
		\le 2d\exp\pa{-n \frac{\{\tau(\nu)\}^2}{\Mr{B2-bd}^2 (u_p)^2}}
		\\
		&\le 
		2d \exp(- \log(2pd / \nu)) = \nu.
	\end{align*}
\end{proof}

\begin{lemma}\label{le:large-dev}
	Suppose Assumption \ref{cond:Bernstein-moments} holds.
	Define an event 
	\[
	\mc{B} = \set{\frac{1}{n} \Norm{\sum_{i=1}^n \varepsilon_i}_2 \le \frac{\lambda}{\Mr{B2-bd} u_p n} }.
	\]
	Choosing $\lambda$ as \eqref{eq:lambda}, we obtain
	\[
	\pr(\mc{B}^c) \le \nu.
	\]
\end{lemma}

\begin{proof}
	Lemma \ref{le:bernstein-ineq} implies
	\begin{align*}
		\pr(\mc{B}^c) 
		&\le 
		\pr\pa{\frac{1}{n} \Norm{\sum_{i=1}^n \varepsilon_i}_2 > \frac{\lambda}{\Mr{B2-bd} u_p n}}
		\\
		&\le 
		2d \exp\pa{-n \frac{(2\sigma^2_{\varepsilon,n} \{\tau(\nu)\}^2 + 4b^2 \tau^4(\nu) + 2\sqrt{2} b \sigma_{\varepsilon,n} \tau^3(\nu))/(\Mr{B2-bd}^2 (u_p)^2)}{2\sigma^2_{\varepsilon,n} + (2\sqrt{2} b \sigma_{\varepsilon,n} \tau(\nu) + 4b^2 \{\tau(\nu)\}^2) / (\Mr{B2-bd} u_p)}}
		\\
		&\le 
		2d \exp\pa{-n \frac{\{\tau(\nu)\}^2}{\Mr{B2-bd}^2 d^2} \frac{2\sigma^2_{\varepsilon,n} + 2\sqrt{2} b \sigma_{\varepsilon,n} \tau(\nu) + 4b^2 \{\tau(\nu)\}^2}{2\sigma^2_{\varepsilon,n} + 2\sqrt{2} b \sigma_{\varepsilon,n} \tau(\nu) + 4b^2 \{\tau(\nu)\}^2}}
		\\
		&\le 
		2d \exp(- \log(2d / \nu)) = \nu.
	\end{align*}
\end{proof}

\subsubsection*{Proof of Theorem \ref{th:nonasymptotic-bd-gauss}}
Assume that the event $\mc{A}$ has happened. On the event $\mc{A}$, it follows from Lemma \ref{le:risk-ub} and Lemma \ref{le:large-dev-gauss} that
\begin{align*}
	\fr{\hat{\Sigma} - \Sigma }^2 
	&\le
	\frac{2 u_p \Mr{B2-bd}}{n} \Norm{\sum_{i=1}^n \varepsilon_i}_2 \av{\theta^* - \hat{\theta}}_1
	+
	\frac{2\lambda}{n} (\av{\theta^*}_1 - \av{\hat{\theta}}_1)
	\\
	&\le 
	\frac{2\lambda}{n} \av{\theta^* - \hat{\theta}}_1
	+
	\frac{2\lambda}{n} (\av{\theta^*}_1 - \av{\hat{\theta}}_1)
	\le
	\frac{4\lambda}{n} \av{\theta^*}_1 
	\\
	&\le
	4\sqrt{2} \sigma_{W,n} u_p \av{\theta^*}_1 \sqrt{\frac{\Mr{B2-bd}^2 \log(2d / \nu)}{n}}
\end{align*}

\subsubsection*{Proof of Theorem \ref{th:nonasymptotic-bd}}
We assume that the event $\mc{B}$ has happened and proceed as in the proof of Theorem \ref{th:nonasymptotic-bd-gauss} using the result of Lemma \ref{le:large-dev}. It follows that
\begin{align*}
	\fr{\hat{\Sigma} - \Sigma }^2 
	\le
	\frac{4\lambda}{nd} \av{\theta^*}_1 
	\le
	4\sqrt{2} \sigma_{\varepsilon,n} u_p \av{\theta^*}_1 \sqrt{\frac{\Mr{B2-bd}^2 \log(2pd / \nu)}{n}} +
	8b (u_p)^2 \av{\theta^*}_1 \frac{\Mr{B2-bd}^2 \log(2d / \nu)}{n}.
\end{align*}

%--------------------------------------------------------------------------------
\bibliographystyle{agsm}
\bibliography{CR}

@book{buhlmann2011statistics,
  title={Statistics for high-dimensional data: methods, theory and applications},
  author={B{\"u}hlmann, Peter and Van De Geer, Sara},
  year={2011},
  publisher={Springer Science \& Business Media}
}

@article{tropp2012user,
  title={User-friendly tail bounds for sums of random matrices},
  author={Tropp, Joel A},
  journal={Foundations of computational mathematics},
  volume={12},
  pages={389--434},
  year={2012},
  publisher={Springer}
}

@book{wainwright2019high,
  title={High-dimensional statistics: A non-asymptotic viewpoint},
  author={Wainwright, Martin J},
  volume={48},
  year={2019},
  publisher={Cambridge university press}
}

@article{zhu2012short,
  title={A short note on the tail bound of wishart distribution},
  author={Zhu, Shenghuo},
  journal={arXiv preprint arXiv:1212.5860},
  year={2012}
}

@article{zou2017covariance,
  title={Covariance regression analysis},
  author={Zou, Tao and Lan, Wei and Wang, Hansheng and Tsai, Chih-Ling},
  journal={Journal of the American Statistical Association},
  volume={112},
  number={517},
  pages={266--281},
  year={2017},
  publisher={Taylor \& Francis}
}

@article{lan2018covariance,
  title={Covariance matrix estimation via network structure},
  author={Lan, Wei and Fang, Zheng and Wang, Hansheng and Tsai, Chih-Ling},
  journal={Journal of Business \& Economic Statistics},
  volume={36},
  number={2},
  pages={359--369},
  year={2018},
  publisher={Taylor \& Francis}
}

@article{xue2012positive,
  title={Positive-definite $\ell_1$-penalized estimation of large covariance matrices},
  author={Xue, Lingzhou and Ma, Shiqian and Zou, Hui},
  journal={Journal of the American Statistical Association},
  volume={107},
  number={500},
  pages={1480--1491},
  year={2012},
  publisher={Taylor \& Francis}
}

@article{friedman2007pathwise,
author = {Jerome Friedman and Trevor Hastie and Holger H{\"o}fling and Robert Tibshirani},
title = {{Pathwise coordinate optimization}},
volume = {1},
journal = {The Annals of Applied Statistics},
number = {2},
publisher = {Institute of Mathematical Statistics},
pages = {302 -- 332},
year = {2007}
}

@article{tseng2001convergence,
  title={Convergence of a block coordinate descent method for nondifferentiable minimization},
  author={Tseng, Paul},
  journal={Journal of optimization theory and applications},
  volume={109},
  pages={475--494},
  year={2001},
  publisher={Springer}
}

@book{rockafellar1997convex,
  title={Convex analysis},
  author={Rockafellar, R Tyrrell},
  volume={11},
  year={1997},
  publisher={Princeton university press}
}

@article{akaike1974new,
  title={A new look at the statistical model identification},
  author={Akaike, Hirotugu},
  journal={IEEE transactions on automatic control},
  volume={19},
  number={6},
  pages={716--723},
  year={1974},
  publisher={Ieee}
}

@article{Fan:2013,
	author = {Fan, Jianqing and Liao, Yuan and Mincheva, Martina},
	title = "{Large Covariance Estimation by Thresholding Principal Orthogonal Complements}",
	journal = {Journal of the Royal Statistical Society Series B: Statistical Methodology},
	volume = {75},
	number = {4},
	pages = {603-680},
	year = {2013},
	month = {08},
	issn = {1369-7412},
	doi = {10.1111/rssb.12016},
	eprint = {https://academic.oup.com/jrsssb/article-pdf/75/4/603/49507152/jrsssb\_75\_4\_603.pdf},
}

@Article{Langfelder:2008,
	author={Langfelder, Peter
	and Horvath, Steve},
	title={WGCNA: an R package for weighted correlation network analysis},
	journal={BMC Bioinformatics},
	year={2008},
	month={Dec},
	day={29},
	volume={9},
	number={1},
	pages={559},
	issn={1471-2105},
	doi={10.1186/1471-2105-9-559}
}

@Article{Langfelder:2007,
	author={Langfelder, Peter
	and Horvath, Steve},
	title={Eigengene networks for studying the relationships between co-expression modules},
	journal={BMC Systems Biology},
	year={2007},
	month={Nov},
	day={21},
	volume={1},
	number={1},
	pages={54},
	issn={1752-0509},
	doi={10.1186/1752-0509-1-54}
}

@article{Tan:2014,
	author  = {Kean Ming Tan and Palma London and Karthik Mohan and Su-In Lee and Maryam Fazel and Daniela Witten},
	title   = {Learning Graphical Models With Hubs},
	journal = {Journal of Machine Learning Research},
	year    = {2014},
	volume  = {15},
	number  = {95},
	pages   = {3297--3331}
}

@article{Anderson:1973,
	author = {T. W. Anderson},
	title = {{Asymptotically Efficient Estimation of Covariance Matrices with Linear Structure}},
	volume = {1},
	journal = {The Annals of Statistics},
	number = {1},
	publisher = {Institute of Mathematical Statistics},
	pages = {135 -- 141},
	year = {1973},
	doi = {10.1214/aos/1193342389}
}

@article{Zwiernik:2017,
	author = {Zwiernik, Piotr and Uhler, Caroline and Richards, Donald},
	title = "{Maximum Likelihood Estimation for Linear Gaussian Covariance Models}",
	journal = {Journal of the Royal Statistical Society Series B: Statistical Methodology},
	volume = {79},
	number = {4},
	pages = {1269-1292},
	year = {2016},
	month = {11},
	issn = {1369-7412},
	doi = {10.1111/rssb.12217},

	eprint = {https://academic.oup.com/jrsssb/article-pdf/79/4/1269/49214628/jrsssb\_79\_4\_1269.pdf},
}

@article{Rahmatallah:2013,
	author = {Rahmatallah, Yasir and Emmert-Streib, Frank and Glazko, Galina},
	title = "{Gene Sets Net Correlations Analysis (GSNCA): a multivariate differential coexpression test for gene sets}",
	journal = {Bioinformatics},
	volume = {30},
	number = {3},
	pages = {360-368},
	year = {2013},
	month = {11},
	issn = {1367-4803},
	doi = {10.1093/bioinformatics/btt687},

	eprint = {https://academic.oup.com/bioinformatics/article-pdf/30/3/360/17344598/btt687.pdf},
}

@article{Oh:2020,
	author = {Oh, Mingyu and Kim, Kipoong and Sun, Hokeun},
	title = {Covariance thresholding to detect differentially co-expressed genes from microarray gene expression data},
	journal = {Journal of Bioinformatics and Computational Biology},
	volume = {18},
	number = {01},
	pages = {2050002},
	year = {2020},
	doi = {10.1142/S021972002050002X},
	note ={PMID: 32336254},
	
	eprint = { 
	https://doi.org/10.1142/S021972002050002X
	
	}
}

@article{Choi:2009,
	author = {Choi, YounJeong and Kendziorski, Christina},
	title = "{Statistical methods for gene set co-expression analysis}",
	journal = {Bioinformatics},
	volume = {25},
	number = {21},
	pages = {2780-2786},
	year = {2009},
	month = {08},
	issn = {1367-4803},
	doi = {10.1093/bioinformatics/btp502},

	eprint = {https://academic.oup.com/bioinformatics/article-pdf/25/21/2780/16890483/btp502.pdf},
}

@article{Bickel:2008b,
	author = "Bickel, Peter J. and Levina, Elizaveta",
	doi = "10.1214/009053607000000758",
	journal = "The Annals of Statistics",
	
	month = "02",
	number = "1",
	pages = "199--227",
	publisher = "The Institute of Mathematical Statistics",
	title = "Regularized estimation of large covariance matrices",

	volume = "36",
	year = "2008"
}

@article{Bickel:2008a,
	author = "Bickel, Peter J. and Levina, Elizaveta",
	doi = "10.1214/08-AOS600",
	journal = "The Annals of Statistics",
	
	month = "12",
	number = "6",
	pages = "2577--2604",
	publisher = "The Institute of Mathematical Statistics",
	title = "Covariance regularization by thresholding",

	volume = "36",
	year = "2008"
}

@article{Rothman:2009,
	author = {Adam J. Rothman and Elizaveta Levina and Ji Zhu},
	title = {Generalized Thresholding of Large Covariance Matrices},
	journal = {Journal of the American Statistical Association},
	volume = {104},
	number = {485},
	pages = {177--186},
	year  = {2009},
	publisher = {Taylor & Francis},
	doi = {10.1198/jasa.2009.0101},
	
	eprint = { 
	https://doi.org/10.1198/jasa.2009.0101
	}	
}

@article{Leng:2011,
	author = {Leng, Chenlei and Li, Bo},
	title = "{Forward adaptive banding for estimating large covariance matrices}",
	journal = {Biometrika},
	volume = {98},
	number = {4},
	pages = {821-830},
	year = {2011},
	month = {09},
	issn = {0006-3444},
	doi = {10.1093/biomet/asr045},

	eprint = {https://academic.oup.com/biomet/article-pdf/98/4/821/17461455/asr045.pdf},
}

@article{Fan:2008,
	title = {High dimensional covariance matrix estimation using a factor model},
	journal = {Journal of Econometrics},
	volume = {147},
	number = {1},
	pages = {186-197},
	year = {2008},
	note = {Econometric modelling in finance and risk management: An overview},
	issn = {0304-4076},
	doi = {https://doi.org/10.1016/j.jeconom.2008.09.017},

	author = {Jianqing Fan and Yingying Fan and Jinchi Lv}
}

@article{Johnstone:2001,
	ISSN = {00905364},

	author = {Iain M. Johnstone},
	journal = {The Annals of Statistics},
	number = {2},
	pages = {295--327},
	publisher = {Institute of Mathematical Statistics},
	title = {On the Distribution of the Largest Eigenvalue in Principal Components Analysis},
	urldate = {2024-02-01},
	volume = {29},
	year = {2001}
}

@article{Ghazalpour:2006,
	doi = {10.1371/journal.pgen.0020130},
	author = {Ghazalpour, Anatole AND Doss, Sudheer AND Zhang, Bin AND Wang, Susanna AND Plaisier, Christopher AND Castellanos, Ruth AND Brozell, Alec AND Schadt, Eric E AND Drake, Thomas A AND Lusis, Aldons J AND Horvath, Steve},
	journal = {PLOS Genetics},
	publisher = {Public Library of Science},
	title = {Integrating Genetic and Network Analysis to Characterize Genes Related to Mouse Weight},
	year = {2006},
	month = {08},
	volume = {2},

	pages = {1-11},
	number = {8},
	
}

@article{Fan:2024,
	author = {Fan, Xinyan and Lan, Wei and Zou,  Tao and Tsai, Chih-Ling},
	title = {Covariance Model with General Linear Structure and Divergent Parameters},
	journal = {Journal of Business \& Economic Statistics},
	volume = {42},
	number = {1},
	pages = {36-48},
	year = {2024},
	publisher = {Taylor & Francis},
	doi = {10.1080/07350015.2022.2142593},
	
	
	eprint = { 
	
	https://doi.org/10.1080/07350015.2022.2142593
	
	
	
	}
	
}

@article{zhang2014confidence,
  title={Confidence intervals for low dimensional parameters in high dimensional linear models},
  author={Zhang, Cun-Hui and Zhang, Stephanie S},
  journal={Journal of the Royal Statistical Society Series B: Statistical Methodology},
  volume={76},
  number={1},
  pages={217--242},
  year={2014},
  publisher={Oxford University Press}
}

@article{van2014asymptotically,
  title={On asymptotically optimal confidence regions and tests for high-dimensional models},
  author={Van de Geer, Sara and B{\"u}hlmann, Peter and Ritov, Ya’acov and Dezeure, Ruben},
  year={2014}
}

@article{zou2006adaptive,
  title={The adaptive lasso and its oracle properties},
  author={Zou, Hui},
  journal={Journal of the American statistical association},
  volume={101},
  number={476},
  pages={1418--1429},
  year={2006},
  publisher={Taylor \& Francis}
}

@article{fan2004nonconcave,
author = {Jianqing Fan and Heng Peng},
title = {{Nonconcave penalized likelihood with a diverging number of parameters}},
volume = {32},
journal = {The Annals of Statistics},
number = {3},
publisher = {Institute of Mathematical Statistics},
pages = {928 -- 961},
year = {2004}
}

@article{fan2001variable,
  title={Variable selection via nonconcave penalized likelihood and its oracle properties},
  author={Fan, Jianqing and Li, Runze},
  journal={Journal of the American statistical Association},
  volume={96},
  number={456},
  pages={1348--1360},
  year={2001},
  publisher={Taylor \& Francis}
}

@Article{Nakatsukasa:2017,
	author={Nakatsukasa, Yuji
	and Soma, Tasuku
	and Uschmajew, André},
	title={Finding a low-rank basis in a matrix subspace},
	journal={Mathematical Programming},
	year={2017},
	month={Mar},
	day={01},
	volume={162},
	number={1},
	pages={325-361},
	issn={1436-4646},
	doi={10.1007/s10107-016-1042-2}
}
%--------------------------------------------------------------------------------

\end{document}